\documentclass[a4paper,USenglish,numberwithinsect,cleveref]{lipics-v2021}

\usepackage{booktabs}
\usepackage{mathtools}
\usepackage{tikz}
\usetikzlibrary{arrows.meta,decorations.pathreplacing,positioning}
\usepackage[paper,cleveref=compatibility]{knowledge}
\knowledgeconfigure{notion,diagnose line=true,label scope=false}
\AddToHook{cmd/intro/before}{\phantomsection}
\AtBeginDocument{\let\label\LaTeXlabel}
\knowledge{conditional membership}{notion}
\knowledge{promise}{notion}
\knowledge{target}{notion}
\knowledge{letter query}{notion}
\knowledge{length}{notion}
\knowledge{promised word}{notion}
\knowledge{target value}{notion}
\knowledge{decision tree}{notion}
\knowledge{output}{notion}
\knowledge{depth}{notion}
\knowledge{correct}{notion}
\knowledge{query complexity}{notion}
\knowledge{partial word}{notion}
\knowledge{completion}{notion}
\knowledge{partial word shatters}{notion}
\knowledge{shattered pair}{notion}
\knowledge{transcript}{notion}
\knowledge{cost}{notion}

\knowledge{constant regime}{notion}
\knowledge{logarithmic regime}{notion}
\knowledge{linear regime}{notion}
\knowledge{regular language}{notion}
\knowledge{semigroup}{notion}
\knowledge{monoid}{notion}
\knowledge{monoid morphism}{notion}
\knowledge{identity}{notion}
\knowledge{idempotent element}{notion}
\knowledge{idempotent power}{notion}
\knowledge{recognizes}{notion}
\knowledge{mutual presentation}{notion}
\knowledge{compatible with multiplication}{notion}
\knowledge{contextual guard relation}{notion}
\knowledge{$s$-contextual matched-guard relation}{notion}
\knowledge{contextual matched-guard relation}{notion}

\knowledge{tolerance}{notion}
\knowledge{tolerance closure}{notion}
\knowledge{k-equal-length equivalence}{notion}
\knowledge{equal-length relation}{notion}
\knowledge{relation shatters}{notion}
\knowledge{guard relation}{notion}
\knowledge{guard tolerance}{notion}
\knowledge{matched-guard relation}{notion}
\knowledge{matched-guard tolerance}{notion}
\knowledge{active interval}{notion}
\knowledge{block query}{notion}
\knowledge{interval}{notion}
\knowledge{big interval}{notion}

\knowledge{zero}{notion}
\knowledge{diagonal relation}{notion}
\knowledge{diagonal pair}{notion}

\knowledge{nondeterministic finite automaton}{notion}
\knowledge{deterministic finite automaton}{notion}
\knowledge{epsilon-transition}{notion}
\knowledge{transition-relation morphism}{notion}
\knowledge{Boolean matrix}{notion}

\newcommand{\apxmode}{append}
\usepackage[appendix=\apxmode,bibliography=common]{apxproof}
\renewcommand{\appendixproofname}[2]{Proof of #1~#2}

\newtheoremrep{theorem}{Theorem}[section]
\newtheoremrep{lemma}[theorem]{Lemma}
\newtheoremrep{claim}[theorem]{Claim}
\newtheoremrep{proposition}[theorem]{Proposition}
\newtheoremrep{corollary}[theorem]{Corollary}

\crefname{example}{Example}{Examples}
\Crefname{example}{Example}{Examples}
\crefname{claim}{Claim}{Claims}
\Crefname{claim}{Claim}{Claims}
\crefname{appendix}{Appendix}{Appendices}
\Crefname{appendix}{Appendix}{Appendices}

\newcommand{\N}{\mathbb{N}}

\newcommand{\Idem}{\operatorname{Idem}}
\newcommand{\cmgrel}{\mathrm{cmg}}
\newcommand{\cgrel}{\mathrm{cg}}
\newcommand{\grel}{\mathrm{g}}
\newcommand{\mgrel}{\mathrm{mg}}
\newcommand{\blank}{\square}
\newcommand{\taulen}{\mathrel{\tau_{\mathrm{len}}}}
\newcommand{\taulenk}{\mathrel{\tau_{k,\mathrm{len}}}}
\newcommand{\taulenx}[1]{\mathrel{\tau_{#1,\mathrm{len}}}}
\newcommand{\taug}{\mathrel{\tau_{\grel}}}
\newcommand{\taumg}{\mathrel{\tau_{\mgrel}}}
\newcommand{\calT}{\mathcal{T}}
\newcommand{\const}{\mathrm{const}}
\newcommand{\lin}{\mathrm{lin}}

\title{Conditional Membership in Regular Languages}
\titlerunning{Conditional Membership in Regular Languages}
\author{Antoine Amarilli}
  {Univ. Lille, INRIA, CNRS, Centrale Lille, UMR 9189 CRIStAL, F-59000 Lille, France}
  {antoine.a.amarilli@inria.fr}
  {https://orcid.org/0000-0002-7977-4441}
  {}
\author{Sebastien Felipe Labbe}
  {Univ. Lille, CNRS, INRIA, Centrale Lille, UMR 9189 CRIStAL, F-59000 Lille, France}
  {sebastien.labbe@univ-lille.fr}
  {https://orcid.org/0009-0007-1790-7820}
  {}
\author{Charles Paperman}
  {Univ. Lille, CNRS, INRIA, Centrale Lille, UMR 9189 CRIStAL, F-59000 Lille, France}
  {charles.paperman@univ-lille.fr}
  {https://orcid.org/0000-0002-6658-5238}
  {}
\authorrunning{A. Amarilli, S. F. Labbe, and C. Paperman}
\Copyright{Antoine Amarilli, Sebastien Felipe Labbe, and Charles Paperman}
\category{Track B: Automata, Logic, Semantics, and Theory of Programming}
\nolinenumbers
\hideLIPIcs

\begin{document}
\maketitle

\begin{abstract}
  We study a formulation of regular language membership as a promise problem, which we dub \emph{conditional membership}: we fix a \emph{promise} language $P$ and a \emph{target} language $T \subseteq P$, we are given a word $w$ which is promised to belong to~$P$, and we must determine whether $w$ belongs to~$T$.
  We study the deterministic worst-case query complexity of this task, when the length $n$ of~$w$ is known in advance and where each query examines the $i$-th position of~$w$ for some $1 \leq i \leq n$ and returns the $i$-th letter $w_i$ of~$w$. Our goal is to characterize the asymptotic complexity 
  as a function of the word length~$n$, depending on the fixed languages $P$ and $T$.

  Our main result is to show that the complexity of this task admits a trichotomy between three regimes: $O(1)$ queries, $\Theta(\log n)$ queries, and $\Theta(n)$ queries. We further show that the trichotomy criterion is decidable, namely, 
  it is PSPACE-complete to decide which regime applies, given nondeterministic automata $A_P$ and $A_T$ respectively representing~$P$ and~$T$. %
\end{abstract}

\vskip\topmattervskip\baselineskip\noindent\noindent\textcolor{lipicsGray}{\fontsize{9}{12}\sffamily\bfseries Usage of AI\enskip} {\fontsize{9}{12}\selectfont 
 AI tools assisted with drafting, editing, proofreading, background research, 
 proposing some technical arguments in discussions with the authors (in particular the lower bound of \cref{thm:effective-classification}),
 and reviewing the ideas and proofs, including checks for gaps, missing assumptions, edge cases and counterexamples, and suggestions for simpler proofs. The authors retain full responsibility for the final manuscript.

\clearpage

\section{Introduction}\label{sec:introduction}

We study the standard question of deciding whether an input word $w$ belongs to a
fixed regular language~$T$. The task is well-known to admit a linear-time algorithm,
by feeding~$w$ to an automaton representing~$T$. Further,
this bound is generally optimal
if we must read the entire input~$w$. However, in settings where $w$ is stored in an array and
we can examine any cell in constant time, then there are cases where we can
decide membership more efficiently. For instance, with the target language $T =
\Sigma^* a$, we can decide membership in constant time, simply by examining the
last cell. Likewise, for
the language $T = a(\Sigma\Sigma)^*b ~|~ c(\Sigma\Sigma)^*\Sigma d$, we can decide
membership simply by examining the first and last letter of the word, if we
assume that the length of the word (and hence its parity) is known in advance. Thus, there
are regular languages for which we can solve membership with just $O(1)$
\emph{queries}. A characterization of 
such languages is given
in~\cite{AaronsonGrierSchaeffer2019}\footnote{The main focus
of~\cite{AaronsonGrierSchaeffer2019} is
regular language membership in a quantum model, which we do not investigate.}, together with a proof
that all other regular languages require linear query complexity, i.e.,
$\Theta(n)$ queries on a word of length~$n$. 

However, in some cases we may be able to test membership to such languages more
efficiently if we can leverage 
some hypotheses on the word~$w$. For instance, we may already know that $w$ is a
valid Unicode string; or more generally we may have been given the information that $w$ belongs
to some language and want to test membership to another. This motivates us to
study regular membership posed as a promise problem~\cite{goldreich2006promise}, which we dub
\emph{conditional membership}. Specifically, we fix a target language $T$
together with a promise language $P$, we are given a word $w$ which is promised
to lie in~$P$, and we wish to determine whether $w$ belongs to~$T$ (the
algorithm can return an arbitrary answer if we try to run it on a word
outside of~$P$).

To our knowledge, this natural conditional membership problem has
not been studied for regular languages, despite its proximity to many similar
questions. The one
exception is the conclusion of~\cite{AaronsonGrierSchaeffer2019} which mentions
it as an open problem, and
already highlights an intriguing phenomenon: with the 
target language $T_0 \coloneq (\Sigma\Sigma)^* (01) \Sigma^*$ asking for an
occurrence of~$01$ at an odd position, and the 
language $P_0 \coloneq 0^*1^*$ asserting that the input is sorted, they notice
that we can decide membership to~$T_0$ conditional to~$P_0$ using binary search,
i.e., with an intermediate regime of $O(\log n)$ queries.
This seems  not to have been investigated further %
since~2019, leaving open the question of whether other regimes are possible (beyond
constant, logarithmic, and linear), and how to characterize the language pairs
that achieve them.

\subparagraph*{Contributions.}
In this paper, we give an effective trichotomy on the query complexity of
conditional membership for arbitrary pairs of a regular target language~$T$ and a
regular promise language~$P$.
We assume for simplicity that $T \subseteq P$ throughout, we assume that the
length~$n$ of the input word $w$ is known in advance, and we consider the model
where $w$ can be accessed by queries that determine the $i$-th letter $w_i$ for
each $1 \leq i \leq n$ (i.e., random-access to the cells of an array that stores~$w$).
We study the
\emph{worst-case query complexity}, i.e., the complexity for length~$n$ is the
minimum height of a decision tree that correctly classifies membership to~$T$ on
input words of length~$n$ that belong to~$P$. To avoid periodic phenomena, we
regularize the complexity by making it monotone, i.e., we study the asymptotic
regime as~$n$ grows of the query complexity obtained by taking the maximum
optimal decision tree height across all lengths  $0 \leq n' \leq n$.

Our trichotomy result is then the following, with the three cases exemplified
in \cref{tab:examples}:

\begin{theorem}
  For any fixed regular language pair $T \subseteq P$, the query complexity of
  conditional membership to~$T$ assuming~$P$ follows one of three regimes:
  \begin{itemize}
    \item \textbf{Constant}, i.e., $O(1)$ queries
    \item \textbf{Logarithmic}, i.e., $\Theta(\log n)$ queries
    \item \textbf{Linear}, i.e., $\Theta(n)$ queries
  \end{itemize}
  Further, the trichotomy criterion is effective: given nondeterministic
  automata $A_P$ and $A_T$ that accept $P$ and~$T$ respectively, it is
  PSPACE-complete to decide
  which case applies.
\end{theorem}

To prove this result, we use the framework of algebraic automata theory. Specifically,
the
characterization is based on a so-called \emph{mutual presentation} of~$P$
and~$T$, i.e., an arbitrary finite monoid $M$ that recognizes them both. Because
the regime only depends on the choice of language pair, any choice of mutual
presentation will give the same regime. The main challenge that we face is
that conditional membership is a promise problem parameterized not by a single
language, but by a \emph{language pair}. For this reason, usual tools such as
language varieties are not immediately applicable, and new methods are needed: the
present paper can also be seen as first steps towards the development of such
algebraic techniques in a new context, which may be relevant beyond the specific
problem that we study.

At a high level, our proof proceeds by defining two \emph{tolerances} on the
mutual presentation~$M$: tolerances are similar to congruences on monoids but
forego the transitivity requirement~\cite{Radelecki1991}. The tolerances
intuitively relate pairs of elements that may be hard to distinguish, i.e.,
which would require at least some asymptotic number of queries to be told apart. The first
tolerance is the \emph{guard tolerance} which relates pairs of the form $(exf,
eyf)$, where $e$ and $f$ are two idempotents with nonempty preimages and where
$x$ and $y$ are two elements with equal-length preimages. The
\emph{matched-guard tolerance} is a refinement of this relation obtained by
relating pairs of the form $(exe, eye)$ with \emph{one common} surrounding
idempotent $e$ (still having 
a nonempty preimage), and with different preimage length requirements.
We are interested in whether these tolerances \emph{shatter}
the promise-target pair, i.e., whether they relate two elements whose 
preimages are in~$P$ but where one of them has preimages in~$T$ and the other
does not. Intuitively, when a tolerance is not shattered, then it justifies that
a conditional membership algorithm can safely avoid querying some parts of the
word. And whenever a tolerance is shattered, then the witnessing
pair can be used to design a lower bound justifying that some number of queries
is necessary to distinguish between the two elements (i.e., to locate the
equal-length middle part that can achieve either $x$ or $y$ depending on its
contents, and which may be ``hidden'' in a long
sequence of idempotent factors).

Our specific technical upper bounds are then the following. First, when the guard
tolerance is not shattered, we show that the language pair has constant query
complexity. This is not too difficult and is simply done by querying a constant-length
prefix and suffix of the input word, similarly to the constant query complexity
case of regular language membership without a promise. Second, when the
matched-guard tolerance is not shattered, we show the language pair has query
complexity $O(\log n)$. The idea for this upper bound is a binary search,
similar to the example from~\cite{AaronsonGrierSchaeffer2019} mentioned earlier, but made much
more complicated by the fact that we need to perform constantly many binary
searches with some elements playing different roles in different parts of the
word.
In both upper bounds, we rely on the fact that sufficiently
long word infixes are guaranteed to contain idempotent factors as a consequence
of Ramsey's theorem (see~\cite[Section~1.2]{Jecker2021}), so we always query
word blocks that are long enough to contain idempotents (to be used for pairs in
our tolerances).

As for lower bounds, our methodology is the same for both: when one of the
tolerances shatters the promise-target pair $P,T$, we use a shattering pair in
the tolerance as a building block to design a reduction. We always reduce from
two concrete promise-target pairs, given as the two bottom rows of \cref{tab:examples}. The
first one,
with $T = a^* c b^*$ and $P = a^* (c|d) b^*$,
has a logarithmic query complexity lower bound, intuitively because we
must locate the middle part.
The second one,
with $T = a^* c a^*$ and $P = a^* (c|d) a^*$,
has a linear lower bound.
To reduce from these languages, we design general-purpose
algebraic lemmas showing that intractability is intuitively preserved under some
operations: moving to a more specific promise-target specification, taking
quotients, taking inverse image by length-multiplying morphisms, and performing
a form of concatenation.

We last show that all our algebraic techniques can be applied in polynomial time on
the mutual presentation of the languages, so that the trichotomy criterion can
be decided in PSPACE given automata representing the promise and target
languages, going via an implicit representation of the mutual presentation;
and show PSPACE-hardness of telling constant and linear languages apart. We next review some related work, and give the outline of the paper.

\subparagraph*{Related work.}
Our result uses algebraic automata theory, which is a
common methodology to classify the complexity of problems
parameterized by regular languages, provided that the tractable or intractable
cases are closed under some algebraic operations (e.g., form a variety).
Examples of the methodology include \cite{BCST92} (for circuit complexity), \cite{TessonTherien2005} (for
\emph{communication complexity}),
\cite{AmarilliJachietPaperman2021} (for
\emph{dynamic membership}),
\cite{martens2023trichotomy} (for \emph{trail queries on graph databases}),
\cite{DBLP:journals/theoretics/GanardiHLMS25} (for \emph{sliding window
membership})
and \cite{AmarilliLabbePaperman2026} (for
\emph{out-of-order membership}). However, in all these works, the studied
problem is
parameterized by one single language, unlike ours. 
We also note that our work classifies arbitrary regular language pairs and not
monoids, i.e., we do not assume the presence of a \emph{neutral letter}: this is
achieved by some of the works
above~\cite{AmarilliJachietPaperman2021,martens2023trichotomy,DBLP:journals/theoretics/GanardiHLMS25}
but not others~\cite{TessonTherien2005,AmarilliLabbePaperman2026}.

Other works have investigated the complexity of membership to regular languages
in the so-called \emph{property testing} model. This is somewhat reminiscent of
our promise problem: the access model is the same, and the tester must accept
for words in the target and reject words that are $\epsilon$-far from
the target; otherwise the tester can answer arbitrarily.
This problem was introduced in~\cite{AlonKrivelevichNewmanSzegedy2001}, with
recent
work showing a trichotomy on the query complexity of
property testing depending on the target
language~\cite{BathieFijalkowMascle2025}. However, they do not cover our
setting where the target and promise language are arbitrary regular languages
(rather than being defined in terms of the distance to a single language).

Our problem is also reminiscent of 
the question of regular language \emph{separation}~\cite{PlaceZeitoun2016}: given two disjoint languages, the objective is
to find a separator (often in a restricted language class) that accepts the
first and rejects the second. However, the goal of finding a separator
is different from our task of efficient conditional membership testing.
Furthermore, conditional membership
does not boil down to simply rewriting the promise language $P$ and target
language $T$ to a
single separator regular language $T'$ on which to test membership (unconditionally).
Indeed, if this were the case, then conditional membership would only admit constant and
linear complexity regimes, but the example of~\cite{AaronsonGrierSchaeffer2019}
already shows that conditional membership also has a logarithmic regime.

One setting in which conditional membership for regular languages has been
significantly studied
is through the lens of 
\emph{state complexity}, i.e., determining the minimum number of states
needed in an automaton that must separate one language from
another~\cite{geffert2015classical}. 
The best known instance of this setting is
perhaps the \emph{separating words problem}~\cite{goralvcik1986discerning}, 
where the two languages to be separated are singleton languages.
This relates to works that have investigated three-valued
automata variants (which may accept, reject, or return a ``don't know'' answer)
in the context of verification: these are called \emph{separating DFAs}
in~\cite{chen2009learning} and \emph{3DFAs} in~\cite{dell2024dfaminer}.
In all these works, however, the focus is on optimizing the state complexity of separating
automata, unlike our work which focuses on query complexity and is
independent of the representation of the languages.

Perhaps the most closely related work to ours is the work 
of Moshkov~\cite{Moshkov2000}. He fixes a promise language~$P$, and studies
the following problem: given an unknown
word~$w$ which is promised to belong to~$P$, study how many queries are needed
to determine $w$ \emph{exactly} in the worst case. His work also shows a trichotomy
on~$P$ with three different regimes 
(constant, logarithmic, and linear), with a criterion defined in terms of an
automaton for~$P$. While the problems are different, the upper bounds
of~\cite{Moshkov2000} transfer to our setting in the following sense:
whenever a promise language $P$ enjoys a sublinear
upper bound thanks to~\cite{Moshkov2000}, then conditional membership for
\emph{any} target language~$T$ relative to $P$ enjoys the same upper bound. This is
because determining $w$ exactly is sufficient to determine whether
$w$ belongs to any target language. However, we are not aware of a reduction in
the other direction, because the focus of~\cite{Moshkov2000} is on (non-fixed)
single-word target languages. Further, and more importantly,
the classification of~\cite{Moshkov2000} only applies to the promise language,
whereas our result characterizes the complexity of conditional membership as a
function of both languages.

\begin{table}[t]
\centering
\small
\caption{Examples of conditional membership problems.}
\label{tab:examples}
\begin{tabular}{@{}p{0.18\textwidth}p{0.27\textwidth}p{0.23\textwidth}p{0.20\textwidth}@{}}
\toprule
\raggedright Query complexity & Promise $P$ & Target $T$ & Defined in\\
\midrule
$O(1)$ & $(c\mid d)a^*$ & $ca^*$ & \Cref{exa:constant}\\
$\Theta(\log n)$ & $a^*(c\mid d)b^*$ & $a^*cb^*$ & \Cref{exa:promise}\\
$\Theta(n)$ & $a^*(c\mid d)a^*$ & $a^*ca^*$ & \Cref{exa:linear}\\
\bottomrule
\end{tabular}
\end{table}

\subparagraph*{Paper outline.}
We give prerequisites, formally define our problem, and state our main result in
\cref{sec:preliminaries}. The proof then spans
\cref{sec:constant,sec:logarithmic}. In 
\cref{sec:constant}, we define the \emph{guard tolerance}, show the constant
upper bound when the guard tolerance is not shattering, and show the logarithmic
lower bound when it is shattering. In 
\cref{sec:logarithmic}, we define the \emph{matched-guard tolerance}, show the
logarithmic
upper bound when the matched-guard tolerance is not shattering, and show the
linear
lower bound when it is shattering. 
Last, in \cref{sec:effective}, we show that the trichotomy criterion is
effective, by giving our PSPACE-completeness result. We conclude and propose
future work directions in \cref{sec:conclusion}. For lack of space, most
detailed proofs are deferred to the appendix.

\section{Preliminaries and Main Result}\label{sec:preliminaries}

Write $\N=\{0,1,2,\ldots\}$, $\N_+\coloneq\N\setminus\{0\}$, and $[n]=\{1,\ldots,n\}$ for $n\in\N$, with $[0]=\varnothing$.
All logarithms are base two unless otherwise specified.
We assume the finite alphabet $\Sigma$ has at least two letters, because for $|\Sigma|=1$ knowing the length suffices to determine the word without making queries. Write $\Sigma^*$ for all words over~$\Sigma$, $\varepsilon$ for the empty word, $\Sigma^+\coloneq\Sigma^*\setminus\{\varepsilon\}$, 
$\Sigma^n$ for words of length $n \in \N$. For $w\in\Sigma^n$, its \intro{length} is $|w|=n$, and $w_i\in\Sigma$ denotes its letter at position $i\in[n]$.

We use the usual definitions and notation for \intro[regular language]{regular
languages}, regular expressions, union, concatenation, and Kleene star
as in~\cite{HopcroftMotwaniUllman2001}.
For 
$u\in\Sigma^*$, we write $u^+=uu^*$. For $L\subseteq\Sigma^*$ and $u,v\in\Sigma^*$, the \intro{quotient}
$u^{-1}Lv^{-1}$ is $\{w\in\Sigma^*\mid uwv\in L\}$.

Our \kl{conditional membership} problem will be defined by two \kl[regular language]{regular languages} $T\subseteq P\subseteq\Sigma^*$, with \intro{promise} $P$ and \intro{target} $T$. The unknown \intro{promised word} $w$ belongs to $P$, 
its length $n\coloneq|w|$ is known, and
its \intro{target value} is $1$ if $w\in T$ and $0$ otherwise: we will have to determine it by making \intro[letter query]{letter queries}, where a query at $i\in[n]$ reveals the letter $w_i$. We will introduce \kl{decision trees} to formally define solutions to this problem and thus to measure complexity.

\subparagraph*{Decision trees and partial words.}
For $n\in\N$ a word length, a \intro{decision tree} $\calT$ is a finite rooted tree whose internal nodes are labeled by positions in $[n]$ and have one child edge per letter in $\Sigma$, and whose leaves are labeled $0$ or $1$. 
We require that any two distinct internal nodes on the same root-to-leaf path have different position labels.
Thus, on every root-to-leaf path (called a \intro{branch}), no position is queried twice.
The semantics of~$\calT$ is that it maps every input word $w \in \Sigma^n$ to a value in
$\{0,1\}$ obtained by following a \kl{branch} defined as follows: we start at the root, whenever we reach
a node labeled $i$ then we query $w_i$ and follow the child edge
labeled with the letter $w_i \in \Sigma$ that we receive; and we eventually
reach a leaf whose label is in $\{0,1\}$.
We say that $\calT$ \intro{accepts} $w$ if the label is~$1$, and that it \intro{rejects}~$w$ if the label is~$0$.
We say that~$\calT$
is \intro{correct} for $(P,T)$ at length $n$ if the language $L$ of the words that it accepts is such that $(L \cap P) \cap \Sigma^n = T \cap \Sigma^n$, i.e., the tree must
accept the words of $T$ and reject the words of $P\setminus T$; the answer for words outside of~$P$ is unrestricted.
We simply call $\calT$ \kl{correct} and omit the pair $(P,T)$ and the length~$n$ when clear from context.

The \intro{depth} of~$\calT$
is the maximum number of queries that can be made on a \kl{branch}; the depth of a node~$d$ is the number of queries made before reaching~$d$. Note that in particular, for $n=0$, the tree must always consist of a single leaf at depth zero. 
The \intro{query complexity} $Q_{P,T}\colon\N\to\N$ at $n$ is the maximum across all lengths $0\le k\le n$ of the minimum \kl{depth} of a \kl{correct} \kl{decision tree} at length $k$. 
The reason why we take the maximum over $0 \le k \le n$ is to make the function nondecreasing:
this regularization avoids periodic phenomena that would complicate the study of asymptotic complexity (e.g., if $P$ contains no words of even length, then a correct decision tree of
\kl{depth} 0 exists at every even length).
Of course, we always have $Q_{P,T}(n)\le n$, as witnessed by the tree that always queries every position (i.e., every branch has length~$n$ and nodes at depth $i$ are labeled by position~$i+1$ for each $0\le i<n$).
Note that our definitions focus on depth, not tree size (see~\cite{BuhrmanDeWolf2002} for a broader landscape of complexity measures).
Throughout the paper, 
we will study the semantics of decision
trees in terms of the \emph{partial words} specified by their branches.
Fix a fresh symbol $\blank\notin\Sigma$. For $n\in\N$, a \intro{partial word} of length $n$ is a word $w\in(\Sigma\cup\{\blank\})^n$. A \intro{completion} of~$w$ is a word of~$\Sigma^n$ obtained by replacing every $\blank$ of~$w$ with a letter of $\Sigma$. The \intro{transcript} of a node $d$ in a decision tree~$\calT$ is the partial word containing the answers to queries along the path from the root of~$\calT$ to~$d$ (not including the query that labels~$d$), and $\blank$ elsewhere. 
Note that $\calT$ reaches node~$d$ while reading a word $w \in \Sigma^n$ precisely when $w$ is a completion of the transcript of~$d$.
We say that a partial word $w$ \intro[partial word shatters]{shatters} $(P,T)$
if $w$ has two completions $u,v\in\Sigma^n$ forming a \intro{shattered pair}:
one lies in $T$ and the other in $P\setminus T$. The following easy lemma shows that we can in fact define the correctness of trees simply via shattered pairs:

\begin{lemmarep}
\label[lemma]{lem:insensitive-leaf}
If a \kl{decision tree} $\calT$ for words of length $n\in\N$ is \kl{correct}  for $(P,T)$ at length $n$, then for any leaf $\ell$ of~$\calT$ the
\kl{transcript} of~$\ell$ does not \kl[partial word shatters]{shatter} $(P,T)$. Conversely, if a \kl{decision tree}
$\calT$ ensures that all its leaves have a \kl{transcript} that does not \kl[partial word shatters]{shatter}
$(P,T)$, then we can make $\calT$ \kl{correct} for $(P,T)$ at length $n$
by assigning labels to its leaves.
\end{lemmarep}

\begin{proof}
  We first prove the first part of the claim by establishing the
  contrapositive. Assume 
  that there is a leaf $\ell$ of~$\calT$ whose \kl{transcript} \kl[partial word shatters]{shatters}~$(P, T)$. This means that
  there are two \kl[completion]{completions} $u$ and $v$ of the \kl{transcript} of~$\ell$ such
  that precisely one of them is in $T$ and the other is in $P \setminus T$.
  As both $u$ and $v$ have length~$n$ and both reach $\ell$ when read
  by~$\calT$, the tree gives the same answer on both so it errs on one of them,
  i.e., $\calT$ is not \kl{correct}.

  For the second part of the claim, assume that all leaves of $\calT$ have a
  \kl{transcript} that does not \kl[partial word shatters]{shatter} $(P,T)$. Thus, for each leaf $\ell$, when
  considering the set $L_\ell$ of words that reach $\ell$ (which are precisely the
  \kl[completion]{completions} of the \kl{transcript} of~$\ell$), when restricting to the words of~$P$, there
  are two possibilities: either all belong to~$T$, or none of them do. We label
  $\ell$ by $0$ or $1$ accordingly. If the restriction of~$L_\ell$ to $P$ is empty, then we 
arbitrarily label the leaf with $0$. Letting $\calT'$ be the result of labeling all leaves in this way,
we obtain that $\calT'$ correctly accepts or rejects on any word of~$P$, so $\calT'$ is correct.
\end{proof}

Thus, in subsequent constructions that define decision trees, we can always leave the leaf labels of the tree unspecified:
by \cref{lem:insensitive-leaf}, it suffices to show that no leaf
\kl{transcript} \kl[partial word shatters]{shatters} $(P,T)$ to show that the leaves can be labeled to make the tree correct.

\subparagraph*{Algebraic prerequisites.}
We use \intro[semigroup]{semigroups}, \intro[monoid]{monoids}, \intro[monoid morphism]{monoid morphisms}, \intro[identity]{identity elements}, \intro[idempotent element]{idempotents}, \intro[subsemigroup]{subsemigroups}, \intro[submonoid]{submonoids}, 
\intro[direct product]{direct products},
and \intro{generated submonoids}~\cite[Section~1.1]{Steinberg2016}.
All \kl[semigroup]{semigroups} and \kl[monoid]{monoids} considered in this paper are finite; we omit the qualifier ``finite'' from now on.
For $M$ a monoid,
we write $1_M$ for its \kl{identity}, so $1_Mx=x1_M=x$ for $x\in M$. 
A \intro{zero} of~$M$ is an element $0_M\in M$ such that $0_Mx=x0_M=0_M$ for all $x\in M$.
For $r\in\N$, we write $x^r$ for the $r$-fold product, with $x^0=1_M$.
We define the \intro{diagonal relation} of~$M$ to be $\Delta_M\coloneq\{(x,x)\mid x\in M\}$, and its elements are called \intro[diagonal pair]{diagonal pairs}.
For a semigroup $S$, let $\Idem(S)\coloneq\{e\in S:e^2=e\}$ be the set of \kl[idempotent element]{idempotent elements}
of~$S$.
For $S$ a semigroup,
an \intro{idempotent power} $\omega$ is an integer in $\N_+$ such that $x^\omega$ is idempotent for every $x\in S$; for instance $\omega=|S|!$~\cite[Remark~1.3]{Steinberg2016}.

For a \kl[monoid morphism]{morphism} $\eta\colon \Sigma^*\to M$,
we say that $\eta$ \intro{recognizes}
$L\subseteq\Sigma^*$ if $L=\eta^{-1}(F)$ for some $F\subseteq M$. A language is
\kl[regular language]{regular} if and only if it is recognized by a
surjective \kl[monoid morphism]{morphism} into a
\kl{monoid}~\cite[Chapter~1, Sections~2.2 and~2.4]{Pin1986}.

When considering morphisms to a monoid,
we know by Ramsey's theorem that every sufficiently
long word contains a nonempty contiguous factor whose morphic image
is idempotent~\cite[Section~1.2]{Jecker2021}:

\begin{lemmarep}[\cite{Jecker2021}]
\label[lemma]{lem:idempotent-subfactor}
Let $\eta\colon\Sigma^*\to M$ be a morphism to a monoid of size $m$.
There exists $q\in\N_+$, depending only on $m$,
such that every
$w\in\Sigma^*$ with $|w|\ge q$ admits a factorization $w=uzv$ with
$u,v\in\Sigma^*$ and $z\in\Sigma^+$ such that $\eta(z)$ is idempotent.
\end{lemmarep}

\begin{proof}
By~\cite[Section~1.2]{Jecker2021}, there exists $q\in\N_+$, depending only on $m$, such that every word of length at least $q$ contains a nonempty contiguous factor whose image under $\eta$ is idempotent. Writing such a word as $w=uzv$, with $u,v\in\Sigma^*$ and $z\in\Sigma^+$, establishes the claim.
\end{proof}

We will also rely on the fact that, for a surjective
morphism $\eta\colon\Sigma^*\to M$, every element of $M$
has a preimage of length strictly less than $|M|$:

\begin{lemmarep}
\label[lemma]{lem:bounded-preimage}
Let $\eta\colon\Sigma^*\to M$ be a surjective morphism to a monoid of size $m$.
For every $s\in M$, there is a preimage
$w\in\Sigma^*$ such that $\eta(w)=s$ and $|w|< m$.
\end{lemmarep}

\begin{proof}
  Let $s \in M$. As $\eta$ is surjective, there is a preimage $w \in \Sigma^*$
  such that $\eta(w)=s$. Take $w$ to have the minimal length among such
  preimages. If $|w| < m$, then we are done, so assume by way of contradiction
  that $|w| \geq m$. Let us focus on the prefix of~$w$ of length $m$ by
  writing $w=a_1\cdots a_{m} w'$ with $a_1,\ldots,a_m\in\Sigma$ and $w' \in \Sigma^*$. Now,
  let us consider the 
$m+1$ prefix images
$\eta(\varepsilon), \eta(a_1),\ldots,\eta(a_1\cdots a_{m})$: each of these is an element in
the monoid~$M$ which has size $m$. By the
pigeonhole principle, we can 
choose $0\le i<j\le m$ such that $\eta(a_1 \cdots a_i) = \eta(a_1 \cdots
a_j)$. Thus, 
write $u\coloneq a_1\cdots a_i$,
$z \coloneq a_{i+1}\cdots a_j$, and
$v \coloneq a_{j+1}\cdots a_{m}$, and let us
write $w'' \coloneq u v w'$.  We have $\eta(w'') = \eta(u) \eta(v) \eta(w')$ and
$\eta(u) = \eta(uz)$ so $\eta(w'') = \eta(u) \eta(v) \eta(w') = \eta(uzv)
\eta(w') = \eta(w)$. Thus, $w''$ is also a preimage of $s$ and has strictly
smaller length, a contradiction. This concludes the proof.
\end{proof}

\subparagraph*{Algebraic tools for conditional membership.}
As we work with a language pair $T \subseteq P \subseteq \Sigma^*$, we introduce
the following definition that we use throughout the paper:

\begin{definition}[Mutual presentation]
A \intro{mutual presentation} of
$T\subseteq P\subseteq\Sigma^*$ consists of a surjective
morphism $\eta\colon\Sigma^*\to M$ into a monoid~$M$, and subsets
$M_T\subseteq M_P\subseteq M$ such that
$P=\eta^{-1}(M_P)$ and $T=\eta^{-1}(M_T)$.  We write it as
$(M,\eta,M_P,M_T)$.
\end{definition}

Such a \kl[mutual presentation]{presentation} always exists: given morphisms $\eta_P\colon\Sigma^*\to N_P$ and $\eta_T\colon\Sigma^*\to N_T$ recognizing $P$ and $T$, respectively, where $N_P,N_T$ are monoids, we combine them
into the morphism
$\eta \colon \Sigma^*\to N_P\times N_T$ 
to the \kl{direct product} of the monoids $N_P$ and~$N_T$
by mapping $w\in\Sigma^*$ to $(\eta_P(w),\eta_T(w))$, and then restrict the codomain of~$\eta$ to its image to ensure that $\eta$ is surjective.

Our criteria for \kl{conditional membership} will be phrased using \emph{tolerances}.
This is a standard notion~\cite{Radelecki1991} which is the analogue of \emph{congruences} but without requiring transitivity:

\begin{definition}[Tolerance]
\label[definition]{def:tolerance}
A \intro{tolerance} on a semigroup $S$ is a reflexive and symmetric relation $R\subseteq S\times S$ that is \intro{compatible with multiplication}: $(x,y),(x',y')\in R$ implies $(xx',yy')\in R$. For a monoid $M$, the \intro{tolerance closure} of $V\subseteq M\times M$ is the least tolerance on $M$ containing $V$.
\end{definition}

Equivalently, the \kl{tolerance closure} of $V$ is the submonoid $V_2$ of
$M\times M$ generated by the pairs of $V$, the converse pairs
$\{(y,x)\mid(x,y)\in V\}$, and the \kl[diagonal pair]{diagonal pairs}. Indeed,
$V_2$ is a tolerance that contains $V$ and that must be contained in every tolerance that contains~$V$.
Seeing the generating pairs as letters, the coordinatewise product
defines a surjective monoid morphism from the free monoid on these letters
onto~$V_2$. By \cref{lem:bounded-preimage}, every pair
in the \kl{tolerance closure} is a product of fewer than $|M|^2$
generators, since $V_2$ has at most $|M|^2$ elements.

When considering
\kl[tolerance]{tolerances}, we will want to know if they
\emph{\kl[relation shatters]{shatter}} the
sets $(M_P,M_T)$:

\begin{definition}
For $M_T\subseteq M_P\subseteq M$, a pair $(p,q)\in M\times M$
is a \kl{shattered pair} for $(M_P,M_T)$ if one of $p,q$ belongs
to~$M_T$ and the other belongs to~$M_P\setminus M_T$.
A relation $R\subseteq M\times M$
\intro[relation shatters]{shatters} $(M_P,M_T)$ if it contains
a \kl{shattered pair} for $(M_P,M_T)$.
\end{definition}

\subparagraph*{Result statement:}
Our main result in this paper is then the following trichotomy, which hinges on
tolerances $\taug$ and $\taumg$ that will be defined in the next sections:

\begin{theorem} 
\label[theorem]{thm:target-trichotomy}

Let $T\subseteq P\subseteq\Sigma^*$ be \kl[regular language]{regular
languages}, and let
$(M,\eta,M_P,M_T)$ be any \kl{mutual presentation} of $T$ and $P$.
We can compute \kl[tolerance]{tolerances} $\taug$ and $\taumg$ on $M$ such that
${\taumg}\subseteq{\taug}$ and such that exactly one of the following three cases occurs:
\begin{enumerate}[(1)]
    \item \textup{\bfseries \intro{constant regime}.}
      The \kl{tolerance} $\taug$ does not
    \kl[relation shatters]{shatter} $(M_P,M_T)$, and $Q_{P,T}(n)=O(1)$.

    \item \textup{\bfseries \intro{logarithmic regime}.}
      The \kl{tolerance} $\taug$ \kl[relation shatters]{shatters}
      $(M_P,M_T)$, but the \kl{tolerance} $\taumg$ does
    not \kl[relation shatters]{shatter} $(M_P,M_T)$, and
    $Q_{P,T}(n)=\Theta(\log n)$.

    \item \textup{\bfseries \intro{linear regime}.}
      The \kl{tolerance} $\taumg$ \kl[relation shatters]{shatters} $(M_P,M_T)$, and
    $Q_{P,T}(n)=\Theta(n)$.
\end{enumerate}
\end{theorem}

We prove \cref{thm:target-trichotomy} in \cref{sec:constant,sec:logarithmic},
before discussing the effectiveness of the trichotomy in \cref{sec:effective}.
\Cref{sec:constant} constructs $\taug$, proves the constant upper bound (\cref{prop:lang-constant-ub}), and then proves the logarithmic lower bound (\cref{prop:lang-guard-lb}). \Cref{sec:logarithmic} constructs $\taumg$, proves the logarithmic upper bound (\cref{prop:lang-log-ub}), then the linear lower bound (\cref{prop:lang-linear-lb}). 
Note that the 
classification is an intrinsic property of the language so it is independent of the chosen \kl{mutual presentation}, even if
the tolerances depend on our choice of presentation.

Throughout the proof in \cref{sec:constant,sec:logarithmic}, for notational convenience, we fix the alphabet $\Sigma$ and
the regular languages
$T\subseteq P\subseteq\Sigma^*$ with a \kl{mutual presentation}
$(M,\eta,M_P,M_T)$, and write $m\coloneq |M|$.
We abuse notation by omitting the pairs $(M_P,M_T)$ and $(P,T)$ when there is no ambiguity.
Thus, we call a \kl{tolerance} or a \kl{partial word} \emph{shattering} or \emph{non-shattering} according to whether it shatters the corresponding pair.
We write $S_\eta \coloneq \eta(\Sigma^+)$ for the subsemigroup of images of
nonempty words.
Note that $S_\eta \subseteq M$, with equality if~$1_M$ has a nonempty preimage
by~$\eta$, and $S_\eta = M \setminus \{1_M\}$ otherwise.
We will often consider the idempotents of~$S_\eta$, and use the following
elementary result:

\begin{lemmarep}
\label[lemma]{lem:common-idempotent-preimage}
There exists a fixed integer $\pi\in\N_+$ such that
$\Idem(S_\eta)\subseteq\eta(\Sigma^\pi)$.
\end{lemmarep}

\begin{proof}
We first choose a nonempty preimage of each idempotent.
The set $\Idem(S_\eta)$ is finite
and nonempty, since $S_\eta$ is a nonempty finite semigroup and
every element has an \kl{idempotent power}.
For each $e\in\Idem(S_\eta)$, choose a word
$u_e\in\Sigma^+$ with $\eta(u_e)=e$.

We then take $\pi\in\N_+$ to be the
least common multiple of the lengths $|u_e|$.
Then $u_e^{\pi/|u_e|}$ has length $\pi$ and image $e$,
by idempotence, so indeed every idempotent of $\Idem(S_\eta)$ has a preimage of
length~$\pi$.
\end{proof}

\section{The Constant Regime}\label{sec:constant}\label{sec:fixed-pair-proof}

We start by establishing the constant/logarithmic boundary of \cref{thm:target-trichotomy}:
the \kl{guard tolerance}~$\taug$ gives a constant upper bound when non-shattering
and a logarithmic lower bound when shattering.
We first define $\taug$. Next, we prove the constant upper bound when it is non-shattering.
We then introduce the \kl{\emph{contextual guard relation}} and prove that it equals $\taug$.
This gives a representation of the pairs of~$\taug$ that we use to prove the logarithmic
lower bound when $\taug$ is shattering.

\subsection{The Guard Tolerance}
\begin{toappendix}
  \subsection{The Guard Tolerance}
\end{toappendix}

To define~$\taug$ we first need the following auxiliary relation (also used
in the next section):

\begin{definition}
\label[definition]{def:k-equal-length-tolerance}
For each $k\in\N_+$, the
\intro[k-equal-length equivalence]{$k$-equal-length equivalence}
$\taulenk$ for~$\eta$ is the relation
$\eta(\Sigma^k)\times\eta(\Sigma^k)$.
Equivalently, for $x,y\in M$, we write
$x\taulenk y$ if there are words $u,v\in\Sigma^k$
with $\eta(u)=x$ and $\eta(v)=y$.
Note that $\taulenk$ is an equivalence relation
on $\eta(\Sigma^k)$.

\label[definition]{def:equal-length-tolerance}
The \intro{equal-length relation} $\taulen$ for~$\eta$
is the union of these equivalence relations over all $k\in\N_+$.
Equivalently, $x\taulen y$ if $x$ and $y$
have nonempty preimages of the same length.
\end{definition}

We show that $\taulen$ is a \kl{tolerance} on $S_\eta$ and is generated by the
first $m^2$ lengths.
It need not be a \kl{tolerance} on $M$: indeed, if $1_M$ has no nonempty
preimage by~$\eta$, then the pair $(1_M,1_M)$ is missing.

\begin{lemmarep}
\label[lemma]{lem:equal-length-tolerance}
The relation $\taulen$ for~$\eta$ is a \kl{tolerance} on $S_\eta$.
Moreover, $\taulen=\bigcup_{i=1}^{m^2}\taulenx{i}$.
\end{lemmarep}

\begin{proofsketch}
  The first claim is routine, and the union
  equality follows by \cref{lem:bounded-preimage}.
\end{proofsketch}

\begin{proof}
We first prove that $\taulen$ is a \kl{tolerance} on $S_\eta$ by checking
reflexivity, symmetry, and compatibility with multiplication.

By definition, we have ${\taulen}\subseteq S_\eta\times S_\eta$.
For reflexivity, fix $x\in S_\eta$ and choose a word $u\in\Sigma^+$ with $\eta(u)=x$.
The pair $u,u$ witnesses that $x\taulen x$.

For symmetry, fix $x,y\in S_\eta$ with $x\taulen y$.
Choose nonempty equal-length preimages $u,v\in\Sigma^+$ of $x,y$.
The pair $v,u$ witnesses that $y\taulen x$.

For compatibility with multiplication, fix $x,y,x',y'\in S_\eta$ with $x\taulen y$ and $x'\taulen y'$. Choose equal-length preimages $u,v\in\Sigma^+$ of $x,y$, and equal-length preimages $u',v'\in\Sigma^+$ of $x',y'$. Then $uu'$ and $vv'$ are nonempty, have equal length and images $xx'$ and $yy'$, so we have $xx'\taulen yy'$.

We now prove the length bound described by the union by showing that every pair
in $\taulen$ is witnessed by a pair of nonempty preimages having a common length
$\leq m^2$, where we
recall that $m$ was defined as $m \coloneq |M|$.
Let $A\coloneq\eta(\Sigma)\times\eta(\Sigma)\subseteq M\times M$ be the alphabet.
Let $N\coloneq{\taulen}\cup\{1_N\}$ be the monoid obtained by adjoining a fresh
identity $1_N\notin M\times M$ to the semigroup $\taulen$.
Define the morphism $\theta\colon A^*\to N$ by $\theta(\varepsilon) \coloneq 1_N$
and $\theta((p,q)) \coloneq (p,q)$ for each letter $(p,q)\in A$.
Thus $\theta$ 
does componentwise multiplication on pairs of elements of $M$.
It is surjective because $\taulen$ consists precisely of images of pairs
of nonempty equal-length words.

Fix $(x,y)\in{\taulen}$.
By \cref{lem:bounded-preimage} applied to $\theta$, there is a preimage
$z\in A^*$ with $\theta(z)=(x,y)$ and $|z|<|N|\le m^2+1$.
Since $(x,y)\ne1_N$, the length $k\coloneq|z|$ of the words in this preimage
satisfies $1\le k\le m^2$.
For each letter of $z\in A^*$, choose a pair of letters of $\Sigma$ with suitable image by $\eta$.
This gives words $u,v\in\Sigma^k$ with $\eta(u)=x$ and $\eta(v)=y$.
Hence $x\taulenx{k}y$, proving ${\taulen}\subseteq\bigcup_{i=1}^{m^2}\taulenx{i}$.
The reverse inclusion follows by definition of $\taulen$, so the claimed equality holds.
\end{proof}

Note that the equal-length relation need not be transitive:

\begin{example}
  \label[example]{exa:nontrans}
  Consider the monoid $M_0=\{0,1,a,a^2\}$, where $1$ is the identity,
  $0$ is a \kl{zero},
  and $a^3=1$.
  Let $\Sigma=\{a,z\}$ and define $\eta(a)\coloneq a$ and
  $\eta(z)\coloneq 0$.
  Then $a\taulenx{1}0$, witnessed by the words $a$ and $z$, and $0\taulenx{2}a^2$,
  witnessed by the words $zz$ and $aa$.
  However, $a$ and $a^2$ are not related by $\taulen$: every word with
  nonzero image consists entirely of $a$'s, and its image is determined
  by its length modulo $3$. Thus $\taulen$ is not transitive.
\end{example}

We use $\taulen$ to relate equal-length factors, and
define the \kl{guard relation} by surrounding them by two (possibly
different) \kl[idempotent element]{idempotents} with nonempty preimages, finally
yielding~$\taug$:

\begin{definition} 
\label[definition]{def:guard-tolerance}
The \intro{guard relation}
$R_{\grel}\subseteq M\times M$ consists of the pairs $(exf,eyf)$
for $e,f\in\Idem(S_\eta)$ and $x,y\in M$ with $x\taulen y$.
We define the \intro{guard tolerance}~$\taug$ as the \kl{tolerance closure}
of~$R_{\grel}$.
\end{definition}

\begin{example}
Continuing \cref{exa:nontrans}, 
let us show that the \kl{guard tolerance} $\taug$ of~$M_0$ is not transitive,
by showing that
$R_{\grel}={\taulen}$ so that ${\taug}={\taulen}$.
Here $S_\eta=M_0$, since even the identity has a nonempty preimage,
$1=\eta(a^3)$. In particular, $1\in\Idem(S_\eta)$.
Taking both guards to be $1$ in the definition of $R_{\grel}$ gives
${\taulen}\subseteq R_{\grel}$.
Conversely, for $e,f\in\Idem(S_\eta)$ and $x,y\in M_0$ with $x\taulen y$,
reflexivity gives $e\taulen e$ and $f\taulen f$, and compatibility with
multiplication gives $exf\taulen eyf$. Thus $R_{\grel}\subseteq{\taulen}$.
By \cref{lem:equal-length-tolerance}, $\taulen$ is already a \kl{tolerance}
on $M_0$,
so taking the \kl{tolerance closure} of
$R_{\grel}={\taulen}$ adds no pairs. Hence ${\taug}={\taulen}$, and indeed
$\taug$ is not transitive.
\end{example}

\subsection{Constant Upper Bound}
\begin{toappendix}
  \subsection{Constant Upper Bound}
\end{toappendix}

We now show that if the \kl{guard tolerance}~$\taug$ is
\kl[relation shatters]{non-shattering}, then the \kl{query complexity} is
constant, because 
we can solve \kl{conditional membership} 
by querying a
prefix and a suffix of bounded length.

\begin{propositionrep} %
\label[proposition]{prop:lang-constant-ub}
If the \kl{guard tolerance} $\taug$ is \kl[relation shatters]{non-shattering}, then
$Q_{P,T}(n)$ is in $O(1)$.
\end{propositionrep}

\begin{proofsketch}
Let $q\in\N$ be the bound from \cref{lem:idempotent-subfactor}.
For each length $n\in\N$, the decision tree queries all letters if $n\le2q$,
and the first and last $q$ letters otherwise, so it has depth at most~$2q$.
A shattering leaf transcript would give a shattered pair of completions
sharing the queried prefix $s$ and suffix $t$, giving a shattered pair in~$\taug$ when
taking an idempotent factor in~$s$ and~$t$.
\end{proofsketch}

\begin{proof}
  We describe how to construct a \kl{correct} \kl{decision tree} for each
  $n \in \N$. Let $q$ be the bound given by \cref{lem:idempotent-subfactor}  for
  the morphism $\eta$; this is a constant independent from~$n$.
  If $n\le 2q$, we simply query all positions. Otherwise, the tree
  $\calT$ queries the first and last $q$ positions.
By \cref{lem:insensitive-leaf}, it suffices to show that no leaf \kl{transcript}
of~$\calT$ 
is \kl[partial word shatters]{shattering}.
By contradiction, assume that $\calT$ has a leaf $\ell$ whose \kl{transcript} has
two \kl[completion]{completions} $a,b\in\Sigma^n$, one in $T$ and one in $P\setminus T$. Write $a = s
\alpha t$
and $b = s \beta t$ with $s, t \in
\Sigma^q$ the prefix and suffix queried in the \kl{transcript},
and $\alpha,\beta\in\Sigma^{n-2q}$.
Applying \cref{lem:idempotent-subfactor} to $s$ and $t$, write
$s=u_\ell u_eu_s$ and $t=u_tu_fu_\rho$, where $u_e,u_f\in\Sigma^+$, the other factors belong to $\Sigma^*$, and $\eta(u_e),\eta(u_f)$ are idempotent.

Since $n>2q$, the words $\alpha$ and $\beta$ are nonempty.
Thus $u_s\alpha u_t$ and $u_s\beta u_t$ are nonempty words
of the same length, so
$\eta(u_s\alpha u_t)\taulen\eta(u_s\beta u_t)$.
As $u_e,u_f$ are nonempty and their images
by~$\eta$ are idempotent, we have
$\eta(u_e) \eta(u_s \alpha u_t) \eta(u_f) \mathrel{R_\grel} \eta(u_e) \eta(u_s \beta u_t)
\eta(u_f)$. 
As $\taug$ is the tolerance closure of~$R_\grel$, we also have 
$\eta(u_e) \eta(u_s \alpha u_t) \eta(u_f) \mathrel{\taug} \eta(u_e) \eta(u_s \beta u_t)
\eta(u_f)$, and as $\taug$ is reflexive we have 
$\eta(u_\ell) \mathrel{\taug} \eta(u_\ell)$ and 
$\eta(u_\rho) \mathrel{\taug} \eta(u_\rho)$. Now, as $\taug$ is a \kl{tolerance}, 
it is \kl{closed under multiplication}, so we
deduce
$\eta(a)\taug\eta(b)$. However, one of $\eta(a),\eta(b)$ belongs
to $M_T$ and the other belongs to $M_P\setminus M_T$, so $\taug$
is \kl[relation shatters]{shattering}, contradicting our assumption.

Thus, at every length, there is a \kl{correct} \kl{decision tree} of depth
at most $2q$, which establishes that $Q_{P,T}(n)=O(1)$.
\end{proof}

\begin{example}
\label[example]{exa:constant}
Let $\Sigma=\{a,c,d\}$, $T_{\const}\coloneq ca^*$, and $P_{\const}\coloneq(c\mid
d)a^*$
as in \cref{tab:examples}.
Querying the first position decides membership in~$T$ conditional to~$P$.
To apply \cref{prop:lang-constant-ub}, we show directly that $\taug$ is
\kl[relation shatters]{non-shattering} under any \kl{mutual presentation}.
Each pair in the \kl{guard relation} $R_{\grel}$ has nonempty preimages sharing their
first letter, since the left idempotent has a nonempty preimage.
Each \kl{diagonal pair} has two identical preimages.
Swapping such pairs and multiplying them preserves the property that their
preimages are either both empty or share their first letter, so every pair in
the \kl{tolerance closure} $\taug$ of~$R_{\grel}$ has preimages 
that are either empty or share their first letter.
If both preimages belong to $P_{\const}$, they are nonempty and their common
first letter determines the same \kl{target value}; hence $\taug$ is
\kl[relation shatters]{non-shattering}.
\end{example}

\subsection{The Contextual Guard Relation}
\begin{toappendix}
\subsection{The Contextual Guard Relation}
\end{toappendix}

Having proved the constant upper bound when the \kl{guard tolerance}~$\taug$
is \kl[relation shatters]{non-shattering}, we now turn to the logarithmic
lower bound when $\taug$ contains a \kl{shattered pair}.
However, 
it is not convenient to take the shattered
pair directly from~$\taug$,
because $\taug$
is defined only implicitly as a
tolerance closure so we have no explicit description of its pairs. We will now show that the \kl{guard tolerance} can be
equivalently expressed 
in a more explicit way, called the \emph{\kl{contextual guard relation}}:

\begin{definition}
\label[definition]{def:cgrel}
The \intro{contextual guard relation}
$R_{\cgrel}\subseteq M\times M$ is defined by
\[
R_{\cgrel}\coloneq
\bigcup_{k\in\N_+}
\bigl\{(\ell exf\rho,\ell eyf\rho)\mid
\ell,\rho\in M,\;
e,f\in\Idem(S_\eta),\;
e,f,x,y\in\eta(\Sigma^k)
\bigr\}
\cup\{(z,z)\mid z\in M\}.
\]
\end{definition}

To understand this definition, note that the non-diagonal pairs in~$R_{\cgrel}$ consist of a
pair $(exf, eyf)$ defined like in the \kl{guard relation} $R_{\grel}$, with
$e,f$ idempotents with a nonempty preimage. However, there are
two changes.
First, we require that both idempotents and middle elements must admit
preimages of the same length~$k \in \N_+$, which strengthens the requirement
$x\taulen y$ from the definition of $R_{\grel}$.
Second, we extend the pair with an arbitrary left and right context $\ell$
and~$\rho$.

We now prove that the \kl{contextual guard relation}~$R_{\cgrel}$
equals the \kl{guard tolerance}~$\taug$ used in the constant upper bound, so we
can take our \kl{shattered pair} in the explicit form of the definition
of~$R_{\cgrel}$:

\begin{lemmarep}
\label[lemma]{lem:guard-closure}
We have ${\taug}={R_{\cgrel}}$.
\end{lemmarep}
\begin{toappendix}
  To show \cref{lem:guard-closure}, we first need a preliminary claim to handle
  the common length requirement added in the definition of the \kl{contextual
  guard relation} relative to the \kl{guard relation}. This length requirement
  can always be handled, and we can always freely add a 
  prefix and suffix:
\begin{claim}[Common-length representation]\label{clm:guard-common-length}
For every $(p,q)\in R_{\grel}$ and all $\ell,\rho\in M$,
we have $(\ell p\rho,\ell q\rho)\in R_{\cgrel}$.
\end{claim}

\begin{proof}
Fix $(p,q)\in R_{\grel}$ and $\ell,\rho\in M$.
By the definition of $R_{\grel}$, there are
$e,f_0\in\Idem(S_\eta)$ and $x_0,y_0\in M$ such that
$x_0\taulen y_0$, $p=ex_0f_0$, and $q=ey_0f_0$.
We construct elements $f,x,y,\rho'\in M$ 
such that $f\in\Idem(S_\eta)$,
$e,f,x,y\in\eta(\Sigma^k)$
for some length $k\in\N_+$, and such that
$(\ell exf\rho',\ell eyf\rho')=(\ell p\rho,\ell q\rho)$.
Together with $e\in\Idem(S_\eta)$, these conditions give a
representation of $(\ell p\rho,\ell q\rho)$ as required by the definition
of $R_{\cgrel}$.

Since $x_0\taulen y_0$, choose $r\in\N_+$ and
$X_0,Y_0\in\Sigma^r$ with $\eta(X_0)=x_0$ and $\eta(Y_0)=y_0$.
By \cref{lem:common-idempotent-preimage}, fix
$\pi\in\N_+$ such that
$\Idem(S_\eta)\subseteq\eta(\Sigma^\pi)$.
Since $e,f_0\in\Idem(S_\eta)$, choose
$E_0,F_0\in\Sigma^\pi$ with
$\eta(E_0)=e$ and $\eta(F_0)=f_0$.

The word $F_0^r$ has length $r\pi\geq r$.
Factor it as $F_0^r=UV$, where $V$ is its suffix of length $r$
and $U$ is its remaining (possibly empty) prefix. In particular,
$U,V\in\Sigma^*$, $|U|=r\pi-r$, and $|V|=r$.
Set $H\coloneq VU$; hence $|H|=r\pi$.
Now set $k\coloneq 2r\pi\in\N_+$ and define
$E\coloneq E_0^{2r}$, $F\coloneq H^2$,
$X\coloneq X_0UH$, and $Y\coloneq Y_0UH$.
We have $|E|=2r\pi=k$ and $|F|=2r\pi=k$.
Also, $|X|=|Y|=r+(r\pi-r)+r\pi=k$.
Thus $E,F,X,Y\in\Sigma^k$.

Define $f\coloneq\eta(F)$, $x\coloneq\eta(X)$,
$y\coloneq\eta(Y)$, and $\rho'\coloneq\eta(V)\rho$.
All four elements belong to $M$.
Since $e$ is idempotent, $\eta(E)=\eta(E_0)^{2r}=e$.
Thus $e,f,x,y\in\eta(\Sigma^k)$.
It remains to prove that $f$ is idempotent and that
$(\ell exf\rho',\ell eyf\rho')=(\ell p\rho,\ell q\rho)$.

Set $a\coloneq\eta(U)$ and $b\coloneq\eta(V)$.
Then $a,b\in M$ and
$ab=\eta(UV)=\eta(F_0^r)=f_0$, since $r\geq1$.
By idempotence of $ab$, we have
$(ba)^3=b(ab)^2a=b(ab)a=(ba)^2$.
Multiplying by $ba$ gives $(ba)^4=(ba)^3=(ba)^2$.
As $H=VU$, it follows that $f=\eta(H^2)=(ba)^2$ and
$f^2=(ba)^4=(ba)^2=f$.
Further, $F\in\Sigma^k$ and $k\geq1$, so $f\in S_\eta$.
Therefore $f\in\Idem(S_\eta)$.

The identity $H=VU$ and the factorization $UV=F_0^r$ give
$UH^3V=(UV)^4=F_0^{4r}$.
Thus $\eta(UH^3V)=\eta(F_0)^{4r}=f_0$,
using $r\geq1$ and idempotence of $f_0$.
From the definitions of $x,y,f,\rho'$ we obtain
$\ell exf\rho'=
 \ell ex_0\eta(UH^3V)\rho=\ell ex_0f_0\rho=\ell p\rho$
and, likewise,
$\ell eyf\rho'=
 \ell ey_0\eta(UH^3V)\rho=\ell ey_0f_0\rho=\ell q\rho$.
The elements $\ell,\rho',e,f,x,y$ and the integer $k$
satisfy all conditions in the definition of $R_{\cgrel}$.
Hence $(\ell p\rho,\ell q\rho)\in R_{\cgrel}$.
\end{proof}
We are now ready to prove \cref{lem:guard-closure}:
\end{toappendix}

\begin{proofsketch}
  We show that $R_{\cgrel}\subseteq{\taug}$ because $\taug$ is reflexive (so it
  contains the diagonal pairs) and it contains the non-diagonal pairs
  because we can extend the pairs of $R_{\grel}$  (which impose weaker length requirements)
  with contexts $\ell$ and $\rho$ (using reflexivity of~$\taug$ and
  compatibility with multiplication).

  Then, we show that pairs of $R_{\grel}$ are also in~$R_{\cgrel}$, even when
  extended with arbitrary contexts $\ell$ and~$\rho$. Here we must enforce the
  stronger length requirements of~$R_{\cgrel}$, which we do in an auxiliary
  lemma where we bring the idempotents to a sufficiently high power and use
  conjugation to synchronize them with the middle parts (possibly up to changing
  the right context).

  Last, we show that $R_{\cgrel}$ is closed under the tolerance axioms: as we
  have established $R_{\grel} \subseteq R_{\cgrel} \subseteq {\taug}$ then the
  minimality of the tolerance closure~$\taug$ concludes. The challenging part is
  to check compatibility with multiplication for two non-diagonal pairs: we can
  simply use the left context and idempotent of the first pair, and the right
  idempotent and context of the second pair.
\end{proofsketch}

\begin{proof}
  We first establish $R_{\cgrel}\subseteq{\taug}$ (step 1), then
$R_{\grel}\subseteq R_{\cgrel}$ (step 2), and finally that
$R_{\cgrel}$ is a tolerance (step 3). The last two facts will yield
${\taug} \subseteq R_{\cgrel}$.

For step~1, we show $R_{\cgrel}\subseteq{\taug}$.
Take $(p,q)\in R_{\cgrel}$.
If $p=q$, this is a \kl{diagonal pair}
and it belongs to $\taug$ by reflexivity.
Otherwise, the definition of $R_{\cgrel}$ gives
$\ell,\rho,x,y\in M$, $e,f\in\Idem(S_\eta)$, and
$k\in\N_+$ such that
$p=\ell exf\rho$, $q=\ell eyf\rho$, and
$e,f,x,y\in\eta(\Sigma^k)$.
In particular, $x\taulen y$, so
$(exf,eyf)\in R_{\grel}\subseteq{\taug}$.
The reflexivity of $\taug$ gives
$(\ell,\ell),(\rho,\rho)\in\taug$.
Since $\taug$ is \kl{compatible with multiplication},
$(p,q)=(\ell exf\rho,\ell eyf\rho)\in{\taug}$.

For step~2, we show $R_{\grel}\subseteq R_{\cgrel}$.
This follows immediately from \cref{clm:guard-common-length} by taking
$\ell\coloneq 1_M$ and $\rho\coloneq 1_M$.

For step~3, we prove that $R_{\cgrel}$ is a tolerance on $M$.
It is reflexive because it contains all \kl[diagonal pair]{diagonal pairs}.
It is symmetric: this is immediate
for \kl[diagonal pair]{diagonal pairs}, and
for non-diagonal pairs exchanging $x$ and $y$ in the definition
preserves every condition and exchanges the two coordinates.

We check \kl[compatible with multiplication]{compatibility with multiplication}.
There are three cases depending on the kind of elements we are multiplying.
First, the product of two
\kl[diagonal pair]{diagonal pairs}
is again a \kl{diagonal pair}.
Second, if exactly one factor is a \kl{diagonal pair} $(d,d)$ with $d \in M$, then
multiplying a pair of the form
$(\ell exf\rho,\ell eyf\rho)$ on the left replaces $\ell$
by $d\ell$, and multiplying it on the right replaces $\rho$
by $\rho d$. In either case the defining conditions remain
satisfied. Hence these products belong to $R_{\cgrel}$.

Now consider the third case of multiplying two non-diagonal pairs
(i.e., pairs $(p,q)$ with $p\ne q$) of $R_{\cgrel}$.
For each $i\in\{1,2\}$, write the $i$-th pair as
$(\ell_i e_i x_i f_i\rho_i,\ell_i e_i y_i f_i\rho_i)$,
where $\ell_i,\rho_i,x_i,y_i\in M$ and
$e_i,f_i\in\Idem(S_\eta)$.
The definition of $R_{\cgrel}$ provides $k_i\in\N_+$ such that
$e_i,f_i,x_i,y_i\in\eta(\Sigma^{k_i})$.
Consequently, $x_i\taulen y_i$ for $i\in\{1,2\}$;
in particular, $x_i,y_i\in S_\eta$.

Set $z\coloneq f_1\rho_1\ell_2e_2\in M$.
Since $f_1\in S_\eta$, it has a nonempty preimage.
Surjectivity of $\eta$ gives preimages in $\Sigma^*$ for
$\rho_1,\ell_2,e_2$. Concatenating them after a nonempty
preimage of $f_1$ gives a nonempty preimage of $z$.
Thus $z\in S_\eta$, and reflexivity of $\taulen$ on $S_\eta$
gives $z\taulen z$.
Compatibility of $\taulen$ with multiplication on $S_\eta$
now gives $x_1zx_2\taulen y_1zy_2$.

The coordinatewise product of the two given pairs is
$(\ell_1e_1(x_1zx_2)f_2\rho_2,
  \ell_1e_1(y_1zy_2)f_2\rho_2)$.
Apply Claim~\ref{clm:guard-common-length} to the pair
$(e_1(x_1zx_2)f_2,e_1(y_1zy_2)f_2)\in R_{\grel}$,
with contexts $\ell_1,\rho_2\in M$.
The claim applies because $e_1,f_2\in\Idem(S_\eta)$
and $x_1zx_2\taulen y_1zy_2$.
The product therefore belongs to $R_{\cgrel}$.
This proves compatibility with multiplication, and hence
$R_{\cgrel}$ is a tolerance on $M$.

Since $R_{\cgrel}$ is a tolerance (by step~3) and since $R_{\cgrel}$ contains $R_{\grel}$
(by step~2), the minimality of $\taug$ gives
${\taug}\subseteq R_{\cgrel}$. Combined with step~1,
this establishes ${\taug}=R_{\cgrel}$.
\end{proof}

\subsection{Logarithmic Lower Bound}
\begin{toappendix}
  \subsection{Logarithmic Lower Bound}
\end{toappendix}

We now turn to the logarithmic lower bound of \cref{thm:target-trichotomy},
showing that a \kl[relation shatters]{shattering} \kl{guard tolerance}
forces at least logarithmic \kl{query complexity}:

\begin{propositionrep}[Logarithmic lower bound]
\label[proposition]{prop:lang-guard-lb}
If the \kl{guard tolerance} $\taug$ is \kl[relation shatters]{shattering}, then $Q_{P,T}(n)$ is in~$\Omega(\log n)$.
\end{propositionrep}

We prove this proposition in the rest of the section. We first illustrate the
intuition of the proof on the specific pair of languages $(P_{\log}, T_{\log})$
from \cref{tab:examples}.
We next formally show
the lower bound on $(P_{\log}, T_{\log})$.
Then, we show general-purpose closure properties which will also be used in the
next section. This will finally allow us to 
lift the lower bound from $(P_{\log}, T_{\log})$ to the general case.

\begin{example}
  \label[example]{exa:promise}
  Let $\Sigma=\{a,b,c,d\}$, $T_{\log}\coloneq a^*cb^*$, and
  $P_{\log}\coloneq a^*(c\mid d)b^*$.
Let $(M,\eta,M_P,M_T)$ be any \kl{mutual presentation} of this pair.
Take the integer $\omega\coloneq m!$ as an \kl{idempotent power} of $M$; for a
word $u\in\Sigma^*$, we write $u^\omega$ to mean $\omega$ concatenated copies of $u$.
Set $e\coloneq\eta(a^\omega)$ and $f\coloneq\eta(b^\omega)$.
Then $e,f\in\Idem(S_\eta)$, and $\eta(c)\taulen\eta(d)$ because $c$ and $d$
have the same length. Hence the images by~$\eta$ of $a^\omega cb^\omega$ and
$a^\omega db^\omega$
are related by~$R_{\grel}$, and therefore by~$\taug$.
These two images form a \kl{shattered pair}, because $a^\omega cb^\omega\in T_{\log}$ and
$a^\omega db^\omega\in P_{\log}\setminus T_{\log}$. This witnesses that $\taug$
 is \kl[relation shatters]{shattering}.

The intuition for why this example does not have constant query complexity is that
the position containing $c$ or $d$, which decides whether a word of $P_{\log}$
belongs to~$T_{\log}$, can be at an arbitrary place in a factor of unbounded length.
To see how its position can move, recall that the blocks $a^\omega$ and $b^\omega$ on its two sides have idempotent images $e$ and $f$.
Repeating one block more times and the other fewer times moves this position and
keeps the total length unchanged.
At the same time, idempotence ensures that the images of the two sides remain $e$ and $f$.
Thus, a correct decision tree must query this position to distinguish $c$ from
$d$, which intuitively requires at least logarithmically many bits of
information on a word of linear size.
\end{example}

\begin{toappendix}
  To show \cref{prop:lang-guard-lb}, we will first show the following lemma:
\end{toappendix}

We now formalize the intuition of \cref{exa:promise} by showing the lower bound
on $(P_{\log},T_{\log})$:

\begin{lemmarep}[Logarithmic lower bound]
\label[lemma]{lem:logarithmic-block-words}
The \kl{conditional membership} problem for
$P_{\log}=a^*(c\mid d)b^*$ and $T_{\log}=a^*cb^*$
has \kl{query complexity} at least logarithmic,
that is, $Q_{P_{\log},T_{\log}}(n)$ is in~$\Omega(\log n)$.
\end{lemmarep}

In fact the query complexity of 
$(P_{\log},T_{\log})$ is precisely $\Theta(\log n)$, as will follow from the
logarithmic upper bound of the next section. For now, let us
sketch the lower bound:

\begin{proofsketch}
Fix $n\in\N_+$ and consider the $n$ \kl{target} words $u_i=a^icb^{n-i-1}$,
for $i\in\{0,\ldots,n-1\}$.
If two such words with different~$i$ reach the same leaf, its \kl{transcript}
cannot contain $c$. Replacing the $c$ in one of these words by
$d$ therefore gives a \kl{completion} in $P_{\log}\setminus T_{\log}$ of the same \kl{transcript},
contradicting \kl[correct]{correctness}. Hence the tree has at least $n$ leaves
and depth at least $\log_{|\Sigma|}n$.
\end{proofsketch}

\begin{proof}
Fix a word length $n\in\N_+$, and let $\calT$ be a \kl{correct} \kl{decision tree} for $(P_{\log},T_{\log})$
at length $n$. We show that the depth of $\calT$ is at least logarithmic in $n$.

We will use words of~$P_{\log} \cap \Sigma^n$ of a specific form.
For every $i\in\{0,\ldots,n-1\}$, let $u_i\coloneq a^icb^{n-i-1}$ and $v_i\coloneq a^idb^{n-i-1}$. The word $u_i$ belongs to $T_{\log}$, whereas the word $v_i$ belongs to $P_{\log}\setminus T_{\log}$. Thus, for every $i$, the pair $(u_i,v_i)$ is a \kl{shattered pair}.

We make the following claim (*): two words $u_i$ and $u_{i'}$ for two
different indices $i \neq
i'$ cannot reach the same leaf $\ell$ of $\calT$. Suppose towards a contradiction that
there exist $i\neq i'$ for which this happens. Equivalently, the \kl{transcript}
$\tau$ of~$\ell$  admits both $u_i$ and $u_{i'}$ as \kl[completion]{completions}.

The \kl{transcript}~$\tau$ cannot contain the letter $c$. Indeed, the word $u_i$ contains $c$ only at position $i+1$, while the word $u_{i'}$ contains $c$ only at position $i'+1$. Since $i\neq i'$, there is no position at which both words contain $c$. Hence, as $\tau$ admits both words as \kl[completion]{completions}, we know that $\tau$ cannot contain $c$.

In particular, position $i+1$, which contains $c$ in $u_i$, is not queried on
the path to~$\ell$. The words $u_i$ and $v_i$ differ only at this position:
$u_i$ contains $c$ there and $v_i$ contains $d$ there. They agree at every other
position. Thus, $v_i$ is also a \kl{completion} of~$\tau$.

Consequently, $\tau$ has a \kl{completion} $u_i\in T_{\log}$ and a \kl{completion} $v_i\in P_{\log}\setminus T_{\log}$, so it
is \kl[partial word shatters]{shattering}. As $\tau$ is the
\kl{transcript} of $\ell$ which is a leaf of~$\calT$, by \cref{lem:insensitive-leaf} this contradicts our assumption that $\calT$ was \kl{correct}.
We conclude that, whenever $i\neq i'$, the words $u_i$ and $u_{i'}$ must reach different leaves, i.e., we have established (*).

Now, there are $n$ words $u_0,\ldots,u_{n-1}$, so $\calT$ must have at least $n$ different leaves. Let $d$ be the depth of $\calT$. Since every internal node of a \kl{decision tree} has at most $|\Sigma|$ outgoing edges, a \kl{decision tree} of depth $d$ has at most $|\Sigma|^d$ leaves. We therefore have $n\leq |\Sigma|^d$, and hence $d\geq\log_{|\Sigma|}n$.

As $\Sigma$ is fixed together with $P_{\log}$ and $T_{\log}$, this shows that
$d$ is in $\Omega(\log n)$.
The \kl{decision tree} $\calT$ was an arbitrary \kl{correct} \kl{decision tree} at length $n$,
so the minimum depth of a \kl{correct} \kl{decision tree} at length $n$ is in $\Omega(\log
n)$. By the definition of \kl{query complexity}, we conclude that
$Q_{P_{\log},T_{\log}}(n)$ is in~$\Omega(\log n)$.
\end{proof}

We transfer the lower bound for $(P_{\log},T_{\log})$ to the
general case using three reductions: inclusion (\cref{lem:query-inclusion}),
quotients (\cref{lem:query-fixed-concatenation}), and inverse length-multiplying
morphisms (\cref{lem:query-inverse-morphism}). Each reduction constructs
a decision tree of no greater depth, using a tree at an input length
that is either unchanged, increased by a fixed constant, or multiplied
by a fixed positive integer, respectively.
We note that these reductions are reminiscent of the closure properties for
varieties of regular languages~\cite{Eilenberg1976,Pin1986}, except that in our
context we are working with language pairs:

\begin{toappendix}
We now state and prove the three reductions used to transfer the logarithmic lower bound: inclusion, quotients, and inverse length-multiplying morphisms.
\end{toappendix}

\begin{lemmarep}[Inclusion]
\label[lemma]{lem:query-inclusion}
Let $ T_1\subseteq P_1\subseteq\Sigma^* 
\text{~and~} 
	T_2\subseteq P_2\subseteq\Sigma^*$
be such that
$
T_1\subseteq T_2 
\text{~and~} 
P_1\setminus T_1\subseteq P_2\setminus T_2.
$
Then, for every $n\in\N$,
$
Q_{P_1,T_1}(n)\le Q_{P_2,T_2}(n).
$
\end{lemmarep}

\begin{proofsketch}
This is simply because, at each length, a \kl{correct} tree for $(P_2,T_2)$ must
also be correct for $(P_1,T_1)$. Indeed, intuitively, the prescribed positive and
negative answers for $(P_1, T_1)$ are a subset of those of $(P_2, T_2)$.
\end{proofsketch}

\begin{proof}
Fix a length $k\in\N$, and let $\calT$ be a \kl{decision tree} that is \kl{correct}
for $(P_2,T_2)$ at length $k$. We claim that the same \kl{decision tree} is
\kl{correct} for $(P_1,T_1)$ at length $k$.

Indeed, every word of $T_1\cap\Sigma^k$ belongs to
$T_2\cap\Sigma^k$, so it is accepted by $\calT$. Likewise, every word of
$(P_1\setminus T_1)\cap\Sigma^k$ belongs to
$(P_2\setminus T_2)\cap\Sigma^k$, so it is rejected by $\calT$.
Hence $\calT$ is \kl{correct} for $(P_1,T_1)$.

Thus, at every length $k$, the minimum depth of a \kl{correct} \kl{decision tree}
for $(P_1,T_1)$ is at most the corresponding minimum depth for
$(P_2,T_2)$. Taking the maximum over $0\le k\le n$ yields
$
Q_{P_1,T_1}(n)\le Q_{P_2,T_2}(n),
$
as required.
\end{proof}

\begin{lemmarep}[Quotient]
\label[lemma]{lem:query-fixed-concatenation}
Let $T\subseteq P\subseteq\Sigma^*$ and let $u,v\in\Sigma^*$. Then, for every $n\in\N$, we have $Q_{u^{-1}Pv^{-1},u^{-1}Tv^{-1}}(n)\le Q_{P,T}(n+|u|+|v|)$.
\end{lemmarep}

\begin{proofsketch}
  For an input $w\in u^{-1}Pv^{-1}$, simulate a \kl{correct} tree for $(P,T)$ on $uwv$. Queries in the fixed prefix $u$ and suffix $v$ are answered for free, and the others query $w$. Since $w\in u^{-1}Tv^{-1}$ exactly when $uwv\in T$, this decides the quotient problem with no greater \kl{depth}.
\end{proofsketch}

\begin{proof}
Fix $k\in\N$, and let $\calT$ be a \kl{decision tree} that is \kl{correct} for
$(P,T)$ at length $|u|+k+|v|$. We construct a \kl{decision tree} $\calT'$ for $(u^{-1}Pv^{-1},u^{-1}Tv^{-1})$ at length $k$ by simulating $\calT$ on $uwv$, where $w\in\Sigma^k$ is the input.

Whenever $\calT$ queries a position belonging to $w$, we query the corresponding
position of $w$ in~$\calT'$. Whenever it queries a position belonging to $u$ or $v$, the answer is already known, so we follow the corresponding edge without making a query in~$\calT'$.

If $w\in u^{-1}Tv^{-1}$, then $uwv\in T$, so $\calT'$ accepts. If $w\in u^{-1}Pv^{-1}\setminus u^{-1}Tv^{-1}$, then $uwv\in P\setminus T$, so $\calT'$ rejects. Hence $\calT'$ is \kl{correct} for $(u^{-1}Pv^{-1},u^{-1}Tv^{-1})$ and has depth at most that of $\calT$.

Thus, for every $k\in\N$, the minimum depth at length $k$ for $(u^{-1}Pv^{-1},u^{-1}Tv^{-1})$ is at most the minimum depth at length $|u|+k+|v|$ for $(P,T)$. Taking the maximum over $0\le k\le n$ yields $Q_{u^{-1}Pv^{-1},u^{-1}Tv^{-1}}(n)\le Q_{P,T}(n+|u|+|v|)$.
\end{proof}

\begin{lemmarep}[Inverse lm-morphism]
\label[lemma]{lem:query-inverse-morphism}
Let $\Gamma$ and $\Sigma$ be finite alphabets, and let $h\colon\Gamma^*\to\Sigma^*$ be a morphism. Assume that there exists an integer $\ell\in\N_+$ such that $|h(a)|=\ell$ for every $a\in\Gamma$. Let $T\subseteq P\subseteq\Sigma^*$, and let $P'=h^{-1}(P)$ and $T'=h^{-1}(T)$. Then, for every $n\in\N$, $Q_{P',T'}(n)\le Q_{P,T}(\ell n)$.
\end{lemmarep}

\begin{proofsketch}
For $w\in P'\cap\Gamma^k$ with $k\in\N$, simulate a \kl{correct} tree for $(P,T)$ on $h(w)$. At the first query in block $h(w_i)$ for $i\in[k]$, query $w_i$ and hence learn the whole block; later queries there are free. The inverse images preserve correctness, depth does not increase, and $|h(w)|=\ell|w|$ yields the stated length scaling.
\end{proofsketch}

\begin{proof}
Fix $k\in\N$, and let $\calT$ be a \kl{correct} \kl{decision tree} for $(P,T)$ at length $\ell k$. We construct a \kl{correct} \kl{decision tree} $\calT'$ for $(P',T')$ at length $k$ whose depth is at most that of $\calT$.

On an input $u=a_1\cdots a_k\in\Gamma^k$, the tree $\calT'$ simulates $\calT$ on $h(u)=h(a_1)\cdots h(a_k)$ in the following way. Whenever $\calT$ queries position $j\in[\ell k]$, then there are unique $i\in[k]$ and $r\in[\ell]$ such that $j=\ell(i-1)+r$, i.e., position $j$ of $h(u)$ is the $r$-th position in the block $h(a_i)$.

If position $i$ of $u$ has not yet been queried, then $\calT'$ queries it. The answer determines $a_i$, and therefore determines the whole block $h(a_i)$ and, in particular, the letter at position $j$ of $h(u)$. If position $i$ has already been queried in an ancestor node in~$\calT'$ (i.e., a position in the same block has already been queried in an ancestor node in~$\calT$), then this is already known. Thus $\calT'$ can in either case follow the same edge as $\calT$.

Therefore $\calT'$ follows the same path that $\calT$ follows on $h(u)$. Moreover, 
our construction ensures that each position of $u$ is queried at most once, and a query is made only when $\calT$ queries a position in the corresponding block. Hence the depth of $\calT'$ is at most the depth of $\calT$.

It remains to verify that $\calT'$ is \kl{correct}. If $u\in P'$, then $h(u)\in P$, and by definition $u\in T'$ if and only if $h(u)\in T$. Since $\calT$ is \kl{correct} for $(P,T)$, the tree $\calT'$ is therefore \kl{correct} for $(P',T')$.

We have thus shown that, for every $k\in\N$, the minimum depth of a \kl{correct} \kl{decision tree} for $(P',T')$ at length $k$ is at most the minimum depth of a \kl{correct} \kl{decision tree} for $(P,T)$ at length $\ell k$. Taking the maximum over $0\le k\le n$ yields $Q_{P',T'}(n)\le Q_{P,T}(\ell n)$, as required.
\end{proof}

With these three reductions we can now show \cref{prop:lang-guard-lb}
from \cref{lem:logarithmic-block-words}:

\begin{proofsketch}
Since $\taug$ is shattering, \cref{lem:guard-closure} gives a shattered pair in $R_{\cgrel}$.
Choose equal-length preimages $E,F,X,Y\in\Sigma^k$, for some $k\in\N_+$,
for its guards and middle elements.
By idempotence, we can choose fixed words $u,v\in\Sigma^*$ so that
$T'\coloneq uE^*XF^*v\subseteq T$ and
$P'\coloneq uE^*(X\mid Y)F^*v$ satisfies $P'\setminus T'\subseteq P\setminus T$.
By inclusion (\cref{lem:query-inclusion}), $(P,T)$ is at least as hard as $(P',T')$.
We use quotients (\cref{lem:query-fixed-concatenation}) to remove $u$ and $v$.
The morphism sending $a,b,c,d$ to $E,F,X,Y$, respectively, multiplies lengths by
$k$, so we use \cref{lem:query-inverse-morphism} and use again
\cref{lem:query-inclusion} to
transfer the logarithmic lower bound
for $(P_{\log},T_{\log})$ from \cref{lem:logarithmic-block-words} to $(P',T')$, and hence to $(P,T)$.
\end{proofsketch}

\begin{toappendix}
  We now have what is needed to prove \cref{prop:lang-guard-lb}:
  \begin{proof}[Proof of \cref{prop:lang-guard-lb}]
    By \cref{lem:guard-closure}, ${\taug}=R_{\cgrel}$. Since $\taug$ is
\kl[relation shatters]{shattering}, we may choose a
\kl{shattered pair} $(p,q)\in R_{\cgrel}$ with $p\in M_T$ and $q\in M_P\setminus M_T$. In particular, $p\neq q$.

By the definition of $R_{\cgrel}$,
there exist $\ell,\rho,x,y\in M$, idempotents
$e,f\in\Idem(S_\eta)$, and an integer $k\in\N_+$ such that
$p=\ell exf\rho$, $q=\ell eyf\rho$, and
$e,f,x,y\in\eta(\Sigma^k)$.
Choose $E,F,X,Y\in\Sigma^k$ representing $e,f,x,y$,
respectively, so $\eta(E)=e$, $\eta(F)=f$, $\eta(X)=x$, and $\eta(Y)=y$,
and $L,R\in\Sigma^*$ with
$\eta(L)=\ell$ and $\eta(R)=\rho$.
Set $u\coloneq LE$ and $v\coloneq FR$.
Since $e,f$ are idempotent, for all $i,j\in\N$ we have
$\eta(uE^iXF^jv)=\ell e^{i+1}xf^{j+1}\rho=p$ and
$\eta(uE^iYF^jv)=\ell e^{i+1}yf^{j+1}\rho=q$.

We now proceed through three consecutive reductions to show that the lower bound for $(P_{\log},T_{\log})$ holds for this problem.

First, we use the inclusion lemma, \cref{lem:query-inclusion}, to show that it
suffices to show our lower bound for $T'\coloneq uE^*XF^*v$ and $P'\coloneq uE^*YF^*v\cup uE^*XF^*v$.
Indeed, $p\in M_T$ and $q\in M_P\setminus M_T$, so $T'\subseteq T$ and $P'\subseteq P$.
The two languages $uE^*XF^*v$ and $uE^*YF^*v$ are disjoint because their words have distinct images $p$ and $q$ under $\eta$. Thus $P'\setminus T'=uE^*YF^*v\subseteq P\setminus T$, so the inclusion lemma applies.

We now remove the fixed prefix $u$ and suffix $v$ using the quotient lemma,
\cref{lem:query-fixed-concatenation}. This establishes that it suffices to show
the lower bound for
$T''\coloneq u^{-1}T'v^{-1}=E^*XF^*$ and $P''\coloneq u^{-1}P'v^{-1}=E^*YF^*\cup E^*XF^*$.

Finally, using the inverse length-multiplying morphism lemma, \cref{lem:query-inverse-morphism}, we reduce this problem to the example languages $T_{\log}$ and $P_{\log}$ from \cref{exa:promise}. The words $E,F,X,Y$ all have length $k$, so define a length-multiplying morphism $h\colon\Gamma^*\to\Sigma^*$, where $\Gamma=\{a,b,c,d\}$, by $h(a)=E$, $h(b)=F$, $h(c)=X$, and $h(d)=Y$. This allows us to apply the inverse length-multiplying morphism lemma.
Every word of $T_{\log}$ maps into $T''$, and every word of $P_{\log}\setminus T_{\log}$ maps into $E^*YF^*$. 
The languages $E^*XF^*$ and $E^*YF^*$ are disjoint: otherwise, adding the prefix $u$ and suffix $v$ 
to a word in the intersection would give a word whose image under~$\eta$ would
have to be both $p$ and $q$, which is impossible because $p \neq q$.
Hence $T_{\log}\subseteq h^{-1}(T'')$ and $P_{\log}\setminus T_{\log}\subseteq h^{-1}(P'')\setminus h^{-1}(T'')$, so the inclusion lemma 
applies to the inverse-image problem.
This inclusion step is needed because in general we do not have $T_{\log} = h^{-1}(T'')$ and 
$P_{\log}\setminus T_{\log} = h^{-1}(P'')\setminus h^{-1}(T'')$, because
the inverse images may contain additional words: for instance, if $E=F=X$, then $h(a)=X\in T''$ although $a\notin T_{\log}$.

Chaining the inequalities given by the reductions, this gives
$Q_{P,T}(kn+|u|+|v|)\ge Q_{P_{\log},T_{\log}}(n)$ for every $n\in\N$. By
\cref{lem:logarithmic-block-words}, the right-hand side is $\Omega(\log n)$.
Since $k>0$, $u$, and $v$ are fixed, monotonicity of $Q_{P,T}$ yields that
$Q_{P,T}(n)$ is in~$\Omega(\log n)$.
  \end{proof}
\end{toappendix}

\section{The Logarithmic Regime}\label{sec:logarithmic}

We continue by establishing the logarithmic/linear boundary of \cref{thm:target-trichotomy}.
For this, we first introduce the \emph{\kl{matched-guard tolerance}}~$\taumg$.
We then prove a logarithmic upper bound when $\taumg$ is
\kl[relation shatters]{non-shattering}.
Next, like in the previous section, we introduce the
\kl[$s$-contextual matched-guard relation]{contextual matched-guard relations}
as a more convenient rephrasing of~$\taumg$ to show a lower bound.
We conclude by showing the linear lower bound when $\taumg$ is shattering.

\subsection{The Matched-Guard Tolerance}
\label{sec:mgt}

We now define the \kl{matched-guard tolerance} $\taumg$. It is similar
to~$\taug$, but requires the same idempotent on both sides
and imposes different length conditions:

\begin{definition}
\label[definition]{def:matched-guard-tolerance}
The \intro{matched-guard relation}
$R_{\mgrel}\subseteq M\times M$ consists of the pairs
$(exe,eye)$ for $e\in\Idem(S_\eta)$ and $x,y\in M$ with $e\taulen x$ and
$e\taulen y$.
We define the \intro{matched-guard tolerance} $\taumg$
as the \kl{tolerance closure} of~$R_{\mgrel}$.
\end{definition}

As claimed in \cref{thm:target-trichotomy}, the matched-guard tolerance 
refines the guard tolerance, so the matched-guard tolerance can only be
shattering if the guard tolerance is:

\begin{lemma}
The \kl{matched-guard relation} is contained in the \kl{guard relation}, and
consequently the \kl{matched-guard tolerance} is contained in the \kl{guard
tolerance}: $R_{\mgrel}\subseteq R_{\grel}$ and ${\taumg}\subseteq{\taug}$.
\end{lemma}

\begin{proof}
Let $(p,q)\in R_{\mgrel}$. By definition, there exist an idempotent
$e\in\Idem(S_\eta)$ and elements $x,y\in M$ such that $(p,q)=(exe,eye)$ and we
have $e\taulen x$ and $e\taulen y$. By symmetry, $x\taulen e$, so compatibility
with multiplication gives $xe\taulen ey$. Since $e(xe)e=exe$ and $e(ey)e=eye$ by
idempotence, the definition of the \kl{guard relation} 
on $(e(xe)e,e(ey)e)$ witnesses that $(p,q)\in
R_{\grel}$. Hence $R_{\mgrel}\subseteq R_{\grel}$. The \kl{tolerance closure} is
monotone, which concludes.
\end{proof}

\begin{example}
\label[example]{exa:log-matched-guard}
We return to $(P_{\log},T_{\log})$ from \cref{exa:promise}, where we showed that $\taug$ is \kl[relation shatters]{shattering}.
We now show that $\taumg$ is \kl[relation shatters]{non-shattering} for a
\kl{mutual presentation} of this pair, so our upcoming logarithmic upper bound will apply.
Consider the monoid $M=\{1,a,b,c,d,0\}$, with $1$ the identity and $0$ a \kl{zero}.
The only nonzero products of $a,b,c,d$ are $a^2=a$, $b^2=b$, $ac=c$, $ad=d$, $cb=c$, and $db=d$.
Let $\eta\colon\Sigma^*\to M$ map each letter to the corresponding element, and set $M_P=\{c,d\}$ and $M_T=\{c\}$.
This gives a \kl{mutual presentation} of $(P_{\log},T_{\log})$.
As before, $(c,d)=(acb,adb)\in R_{\grel}\subseteq{\taug}$, since $a,b\in\Idem(S_\eta)$ and $c\taulen d$.
For $R_{\mgrel}$, however, $\Idem(S_\eta)=\{a,b,0\}$ and $axa\in\{a,0\}$, $bxb\in\{b,0\}$, and $0x0=0$ for all $x\in M$.
Thus every pair of $R_{\mgrel}$ is \kl[diagonal pair]{diagonal} or has a zero component.
The relation $\Delta_M\cup(M\times\{0\})\cup(\{0\}\times M)$ is reflexive and symmetric.
It is also \kl{compatible with multiplication}, because $0$ is a zero.
It is therefore a \kl{tolerance} containing $R_{\mgrel}$, so it contains $\taumg$.
It excludes $(c,d)$, the only \kl{shattered pair} up to symmetry, so $\taumg$ is
\kl[relation shatters]{non-shattering} and ${\taumg}\subsetneq{\taug}$.
\end{example}

\subsection{Logarithmic Upper Bound}
\label{sec:logub}
\begin{toappendix}
\subsection{Logarithmic Upper Bound}
\end{toappendix}

We now show that if the \kl{matched-guard tolerance}~$\taumg$ is
\kl[relation shatters]{non-shattering}, then the \kl{query complexity} is at most logarithmic.
In this case, intuitively, we can solve \kl{conditional membership} with a bounded
number of binary searches over the word to find boundaries between different
regions of the word:

\begin{propositionrep}[Logarithmic upper bound]
\label[proposition]{prop:lang-log-ub}
If the \kl{matched-guard tolerance} $\taumg$ is \kl[relation shatters]{non-shattering}, then $Q_{P,T}(n)=O(\log n)$.
\end{propositionrep}

Before proving this upper bound, we first examine the following example.

\begin{example}
We return to the languages $P_{\log}=a^*(c\mid d)b^*$ and
$T_{\log}=a^*cb^*$ from \cref{exa:log-matched-guard}.
A binary search finds the unique position carrying $c$ or $d$:
an answer $a$ directs the search to the right, while an answer $b$
directs it to the left. Thus the query complexity is at most logarithmic.
As shown in \cref{exa:log-matched-guard}, $\taumg$ is
\kl[relation shatters]{non-shattering} for this pair.
Hence \cref{prop:lang-log-ub}, which we prove below, gives
$Q_{P_{\log},T_{\log}}(n)=O(\log n)$; together with 
\cref{lem:logarithmic-block-words}, this yields
$Q_{P_{\log},T_{\log}}(n)=\Theta(\log n)$.
\end{example}
We now generalize this binary-search intuition to 
the general case.
The overall challenge is that several binary searches may be needed: further, the same idempotent
may occur at different positions, so a single query does not always determine
in which direction to search next.
\Cref{exa:log-two,exa:log-nested} in \cref{app:additional-examples} illustrate
this challenge.
One more immediate problem, however, is that a single query does not always yield an idempotent
factor of the word. Indeed, as illustrated in \cref{exa:log-blocks}, \kl[idempotent element]{idempotent elements} that give us directional information need not be images of single letters and may require queries of entire blocks. This is easily addressed by 
\cref{lem:idempotent-subfactor}: there is a fixed
$q\in\N_+$ such that every factor of at least $q$ consecutive letters
contains such a subfactor. For this reason, we will always query contiguous
blocks of at least
$q$ letters. In fact, for technical reasons, we also require the block length to be a multiple
of the common positive preimage length $\pi$ of all idempotents given by 
\cref{lem:common-idempotent-preimage}.
Thus, for the rest of this section,
we define the block length $K\coloneq q\pi$:
this ensures that
every block of $K$ letters contains a
nonempty factor with idempotent image under $\eta$
(\cref{lem:idempotent-subfactor}), and every
$e\in\Idem(S_\eta)$ has a preimage of length $K$ (by
\cref{lem:common-idempotent-preimage} and idempotence).

\begin{toappendix}
  We first show the auxiliary results given in the main text:
\end{toappendix}

We now explain
how to view words and queries over the block alphabet $\Sigma^K$ which we will
use throughout the proof.
The \emph{decoding} of a word
$W\in(\Sigma^K)^*$ concatenates its blocks into a word over $\Sigma$.
For convenience we will abuse notation and see such words $W$ as words
over~$\Sigma$ (i.e., decoding them implicitly), for instance 
when applying $\eta$ or considering membership in $P$ and $T$.
For $W\in(\Sigma^K)^*$, the length $|W|$ counts blocks, and the
decoding has length $K|W|$.
Similarly, a \kl{partial word} over $\Sigma^K$ is a word over $\Sigma^K \cup
\{\blank\}$ and is (implicitly) decoded to a word over $\Sigma \cup \{\blank\}$
by replacing each blank block with $K$ blanks.

We first introduce \emph{$\taumg$-coherence}, a property of decision trees used
throughout the proof:

\begin{definition} %
\label[definition]{def:log-coherent-tree}
Let $\Sigma$ be an alphabet and let $\calT$ be a
\kl{decision tree} over $\Sigma$. We say that $\calT$ is
\emph{$\taumg$-coherent} if, for every leaf $d$ of $\calT$
with \kl{transcript} $t$ and every two \kl[completion]{completions}
$t_x,t_y$ of $t$, we have $\eta(t_x)\taumg\eta(t_y)$.
For $n\in\N$, the \emph{$\taumg$-coherent query complexity}
$C_\Sigma(n)$ is the maximum, over lengths $0\le k\le n$,
of the minimum depth of a $\taumg$-coherent decision tree
over $\Sigma$ at length $k$.
\end{definition}

The point of $\taumg$-coherence is that it is preserved when decoding decision
trees from $\Sigma^K$ to~$\Sigma$:

\begin{lemmarep} %
\label[lemma]{lem:log-block-tree-expansion}
Let $k,d\in\N$ and let $\calT$ be a $\taumg$-coherent
\kl{decision tree} over $\Sigma^K$ at length $k$ and depth at most $d$.
We can build a $\taumg$-coherent decision tree $\calT'$ over $\Sigma$
at length $Kk$ and depth at most $Kd$.
\end{lemmarep}
\begin{proofsketch}
We simply replace each block query by $K$ letter queries.
\end{proofsketch}

\begin{proof}
We replace each block query by $K$ letter queries. We first construct
$\calT'$,
then verify that $\taumg$-coherence is preserved,
before finally bounding the depth of~$\calT'$.

We first construct $\calT'$. At each internal node of $\calT$
that queries block $i\in[k]$, where $i$ is the queried block position,
replace the node by a complete tree that
queries the positions $K(i-1)+1,\ldots,Ki$ in order. The answers
form a block $A\in\Sigma^K$, the block read at position $i$.
Attach the corresponding child of the
original node at the leaf reached by $A$. Keep all original leaf
labels. Since no block is queried twice along a branch of $\calT$,
and distinct blocks occupy disjoint letter positions, it is also true that
no letter is queried twice along a branch of $\calT'$.

We now prove that coherence is preserved. Let $d'$ be a leaf of $\calT'$ with
transcript $t'$, and let $d$ be the original leaf of~$\calT$ from which it
arose, with transcript $t$. For any two completions $t_x,t_y$ of
$t'$, grouping their letters into blocks gives completions of $t$.
Since $\calT$ is $\taumg$-coherent, we have
$\eta(t_x)\taumg\eta(t_y)$. Hence $\calT'$ is $\taumg$-coherent.

Finally, we bound the depth. Each block query is replaced by
exactly $K$ letter queries. Thus each branch of length at most
$d$ becomes a branch of length at most $Kd$, as claimed.
\end{proof}

Now, as $\taumg$ is \kl[relation shatters]{non-shattering}, then
any $\taumg$-coherent decision tree over $\Sigma$ can be made
\kl{correct} for $(P,T)$ by labeling its leaves, thanks to
\cref{lem:insensitive-leaf}. 
Indeed, two completions $t_x,t_y$ of the same leaf transcript
have $\taumg$-related images, so they cannot form a
\kl{shattered pair}.
This allows us to conclude that our goal of 
building correct decision trees for $P$ and $T$ over~$\Sigma$ reduces to that of building
$\taumg$-coherent decision trees over~$\Sigma^K$:

\begin{lemmarep} %
\label[lemma]{lem:log-block-query-simulation}
If $\taumg$ is \kl[relation shatters]{non-shattering}, then
for all $n\in\N$, we have
$Q_{P,T}(n)\le K\bigl(C_{\Sigma^K}(\lfloor n/K\rfloor)+1\bigr)$,
where $C_{\Sigma^K}$ is the query complexity restricted to
$\taumg$-coherent decision trees over $\Sigma^K$.
\end{lemmarep}
\begin{proofsketch}
  We divide words over $\Sigma$ into blocks of size $K$, with a small suffix of
  length less than~$K$. We query this short suffix, then use
  a $\taumg$-coherent block tree on
the prefix (whose length is a multiple of~$K$), which we obtain by 
\cref{lem:log-block-tree-expansion}.
\end{proofsketch}

\begin{proof}
Fix $n \in \N$: we construct a correct tree $\calT$
over~$\Sigma$ for every length $m\le n$.
Fix a word length $m\in\N$ with $m\le n$,
and write $m=Kk+r$, where $k = \lfloor m/K \rfloor \in\N$ counts full blocks and
$0\le r<K$ is the number of remaining letters.
If $k=0$, query all $m$ letters and label each leaf according to
membership in $T$ for inputs in $P$ (arbitrarily otherwise).
Otherwise, for $k>0$,
let $D\coloneq C_{\Sigma^K}(k)$, the query
complexity restricted to coherent block trees up to length
$k$. 
By the definition of $D$, there exists a
$\taumg$-coherent tree $\calT_k$ over $\Sigma^K$ at length $k$
of depth at most $D$.
Use \cref{lem:log-block-tree-expansion} to translate it to a $\taumg$-coherent
decision tree $\calT'_k$ over~$\Sigma$ at length $Kk$ of depth at most~$KD$.
Then, to build~$\calT$, first query the final $r$ letters using a complete
decision tree, and then
attach a copy of $\calT_k'$ at each leaf.

We next establish correctness of the leaves of~$\calT$. The claim is immediate
for~$k=0$, so assume that $k>0$.
Fix a leaf of $\calT$, and let $u\in\Sigma^r$ be
the suffix revealed by the first $r$ queries. Let $d$ be the
corresponding leaf of $\calT_k'$, let $t$ be its transcript,
and take any two completions
$t_x,t_y$ of~$t$. By $\taumg$-coherence, we have
$\eta(t_x)\taumg\eta(t_y)$, interpreting the block words
through decoding. Since $\eta(u)\taumg\eta(u)$ by reflexivity,
compatibility with multiplication shows that the 
words $t_x u$ and $t_y u$ 
also have images related by~$\taumg$.
As $\taumg$ is non-shattering, these two words cannot lie on
opposite sides of $T$ while both belonging to $P$.
This holds for every two completions of the resulting leaf transcript,
so \cref{lem:insensitive-leaf} gives a correct leaf labeling.

Finally, we bound the depth. If $k=0$, the tree has depth
$m<K$. Otherwise, it has depth at most $r+KD$, where $r<K$.
Hence in both cases the depth is at most $K(D+1)$.
Taking the maximum over all $m\le n$ gives
$Q_{P,T}(n)\le K\bigl(C_{\Sigma^K}(\lfloor n/K\rfloor)+1\bigr)$.
\end{proof}

Thus, by \cref{lem:log-block-query-simulation},
to show the logarithmic
upper bound it suffices to show $C_{\Sigma^K}(k)=O(\log(k+1))$.
Thus, for every $k\in\N$ we must build a $\taumg$-coherent decision tree
$\calT$
over blocks (i.e., on~$\Sigma^K$) that has
logarithmic depth. 
The point of working on blocks is that
any words on $\Sigma^K$ whose first and last letters are all the same must have images
related by $\taumg$; we will use it to show $\taumg$-coherence:
\begin{lemmarep} %
\label[lemma]{lem:log-matched-blocks}
For all $v\in\Sigma^K$ and $X',Y'\in(\Sigma^K)^*$,
we have $\eta(vX'v)\taumg\eta(vY'v)$.
\end{lemmarep}

\begin{proofsketch}
  This relies on the ability to find 
  an idempotent image in the common endpoint
  (with \cref{lem:idempotent-subfactor}), and using idempotence to pad $X'$ and
  $Y'$ to a common length which is also compatible with preimages of the
  idempotents (this relies on
  \cref{lem:common-idempotent-preimage}).
\end{proofsketch}

\begin{proof}

We first fix the two middle words and their lengths. Fix
$v\in\Sigma^K$ and $X',Y'\in(\Sigma^K)^*$.
By symmetry of $\taumg$, up to exchanging $X'$ and $Y'$, assume
$|X'|\le |Y'|$ (when counting blocks, i.e., as words on~$\Sigma^K$). Set
$d\coloneq |Y'|-|X'|\in\N$, the difference in the number of blocks.
We first find an idempotent factor of $v$.
Since $K=q\pi\ge q$, \cref{lem:idempotent-subfactor} gives a
factorization $v=heg$, where $h,g\in\Sigma^*$, $e\in\Sigma^+$, and $\eta(e)$ is
idempotent. In particular, $\eta(e)\in\Idem(S_\eta)$.

Now, we pick a suitable preimage for~$e$.
By \cref{lem:common-idempotent-preimage}, choose $E\in\Sigma^\pi$
with $\eta(E)=\eta(e)$. Since $K=q\pi$ and $q>0$, the word
$E^q\in\Sigma^K$ has image $\eta(E^q)=\eta(e)$ by idempotence.
Put $X\coloneq gX'he(E^q)^d$ and $Y\coloneq gY'he$ in $\Sigma^*$.
Counting letters of $\Sigma$, we have
$|X|=K|X'|+K+dK=K|Y'|+K=|Y|$,
a positive multiple of $K$ and hence of $\pi$.

We next apply the matched-guard relation to equal-length words.
Set $\ell\coloneq |X|/\pi\in\N_+$, the number of blocks of
length $\pi$ in $X$.
The word $E^\ell$ has length $|X|=|Y|$ and image
$\eta(E^\ell)=\eta(e)$ by idempotence. Thus
$\eta(e)\taulen\eta(X)$ and $\eta(e)\taulen\eta(Y)$, so
\cref{def:matched-guard-tolerance} gives
$\eta(eXe)\mathrel{R_{\mgrel}}\eta(eYe)$.
By idempotence of $\eta(e)$, we have
$\eta(e(E^q)^de)=\eta(e)$; this also holds when $d=0$. Hence
$\eta(vX'v)=\eta(hegX'heg)=\eta(hegX'he(E^q)^deg)=\eta(h)\eta(eXe)\eta(g)$ and
$\eta(vY'v)=\eta(hegY'heg)=\eta(hegY'heeg)=\eta(h)\eta(eYe)\eta(g)$.
Since $\taumg$ contains $R_{\mgrel}$, reflexivity and
\kl[compatible with multiplication]{compatibility with multiplication}
give $\eta(vX'v)\taumg\eta(vY'v)$.
\end{proof}

The decision tree $\calT$ over blocks (i.e., letters of~$\Sigma^K$) will proceed
by performing binary search over \emph{intervals}, corresponding to contiguous
ranges of $\blank$ in the transcript of the nodes of~$\calT$, focusing on
so-called \emph{active} intervals on which 
\cref{lem:log-matched-blocks} is intuitively not yet applicable. We define
these notions on an abstract alphabet~$\Gamma$ as they are generic, but we will
always use $\Gamma \coloneq \Sigma^K$:

\begin{definition}[Intervals of a partial word]
\label[definition]{def:log-active-interval}
Let $\Gamma$ be an alphabet and $t$ a \kl{partial word} of length
$k\in\N$ over $\Gamma$. Write
$t=\blank^{k_1}v_1\blank^{k_2}v_2\cdots
\blank^{k_s}v_s\blank^{k_{s+1}}$,
where $s\in\N$, $v_i\in\Gamma$ for $i\in[s]$,
$k_i\in\N$ for $i\in[s+1]$, and
$s+\sum_{i=1}^{s+1}k_i=k$.
For $s=0$, this means $t=\blank^{k_1}$.
For each $i\in\{0,\ldots,s\}$, put
$s_i\coloneq i+\sum_{j=1}^{i}k_j$ and
$e_i\coloneq s_i+k_{i+1}$.
The \intro[interval]{$i$-th interval} of $t$ is the (possibly empty) set of blank
positions $\{s_i+1,\ldots,e_i\}$, of \emph{size} $k_{i+1}$.
The interval is \intro[active interval]{active} if $i=0$, or if
$i\in[s]$ and $v_i$ does not occur among $v_{i+1},\ldots,v_s$,
i.e., $v_i\neq v_j$ for every $j\in\{i+1,\ldots,s\}$.
\end{definition}

Our binary search strategy will focus on active intervals, of which there are
constantly many:

\begin{claim}
\label[claim]{clm:log-active-number}
For every finite alphabet $\Gamma$ and every partial word $t$ over $\Gamma$,
the number of active intervals of $t$ is at most $|\Gamma|+1$.
\end{claim}
\begin{proof}
  It is immediate that the preceding letter $v_i$ of each active interval with
  $i>0$ is unique among active intervals.
Thus, there are at most $|\Gamma|$ such intervals, plus the first interval.
\end{proof}

The point of intervals is that we can use \cref{lem:log-matched-blocks} to show that
\kl[completion]{completions} agreeing on all \kl[active interval]{active intervals}
have images related by $\taumg$, so that intuitively only such intervals must be
queried:

\begin{lemmarep}
\label[lemma]{lem:log-active-completions}

Let $t$ be a \kl{partial word} over $\Sigma^K$ and let
$t_x,t_y$ be two \kl[completion]{completions} of $t$.
If $t_x$ and $t_y$ agree on every \kl{active interval} of $t$, then
$\eta(t_x)\taumg\eta(t_y)$.
\end{lemmarep}

\begin{proofsketch}
We do an induction on the number of non-$\blank$ positions of $t$.
Considering the first position that contains a non-$\blank$ letter $a\in\Sigma^K$, we 
apply \cref{lem:log-matched-blocks} to the factor of the completions of~$t$ from
this position to the last occurrence of~$a$ in~$t$.
The remaining suffix is handled using the induction hypothesis, and we conclude
by compatibility with multiplication.
\end{proofsketch}

\begin{proof}
We prove the invariant by induction on the number $s\in\N$ of
non-$\blank$ positions of an arbitrary partial word $t$ over $\Sigma^K$.
At each step, we compare the factor from the first non-$\blank$ letter to
its last occurrence and then apply induction to the remaining suffix.
Write $k\coloneq|t|$ for the length of $t$.

We establish the base case $s=0$. Then $t=\blank^k$.
Its only interval is active, so any two completions
$t_x$ and $t_y$ of $t$ must be identical,
and the conclusion follows by reflexivity of $\taumg$.

We next set up the inductive step. Let $s\in\N_+$, and assume
the assertion holds for all partial words with fewer than $s$
non-$\blank$ positions.
Recall the factorization
$t=\blank^{k_1}v_1\blank^{k_2}\cdots v_s\blank^{k_{s+1}}$.
For each $i\in\{0,\ldots,s\}$, let $x_i,y_i\in(\Sigma^K)^*$ be the
respective words filling the $i$-th interval of $t$ in $t_x,t_y$;
each of these words has length $k_{i+1}$ in blocks, and is empty
exactly when the $i$-th interval is empty. Thus
$t_x=x_0v_1x_1\cdots v_sx_s$ and
$t_y=y_0v_1y_1\cdots v_sy_s$.
Since the first interval is active, $x_0=y_0$.

We compare the factors with matching endpoints next.
We now focus on the first non-$\blank$ letter $v_1$ of~$t$.
Let $j\coloneq\max\{i\in[s]\mid v_i=v_1\}\in[s]$ be the index of the
last occurrence of $v_1$ in $t$.
Let $p_x,p_y\in(\Sigma^K)^*$ be the respective factors of $t_x,t_y$
from $v_1$ through $v_j$, inclusive, and let
$u_x,u_y\in(\Sigma^K)^*$ be their respective suffixes strictly after
$v_j$. Then $t_x=x_0p_xu_x$ and $t_y=x_0p_yu_y$.
If $j=1$, then we have $p_x=p_y=v_1$, so
$\eta(p_x)\taumg\eta(p_y)$ by reflexivity.
Otherwise, let $X',Y'\in(\Sigma^K)^*$ be the respective factors
strictly between $v_1$ and $v_j$ in $p_x,p_y$.
Since $p_x=v_1X'v_j$, $p_y=v_1Y'v_j$, and $v_j=v_1$,
applying \cref{lem:log-matched-blocks} gives $\eta(p_x)\taumg\eta(p_y)$.

We will now apply the induction hypothesis to the remaining suffix.
Let $t'$ be the suffix \kl{partial word} of $t$ strictly after $v_j$;
when $j=s$, this is $t'=\blank^{k_{s+1}}$.
The words $u_x,u_y$ are \kl[completion]{completions} of $t'$, which
has $s-j<s$ non-$\blank$ positions.
The first interval of $t'$ corresponds to the $j$-th interval of $t$,
which is active because $v_j$ has no later occurrence.
Every other active interval of $t'$ corresponds to an active interval
of $t$, since the occurrences later than its preceding block are
unchanged. Hence $u_x,u_y$ agree on every \kl{active interval} of $t'$.
By the induction hypothesis, $\eta(u_x)\taumg\eta(u_y)$.
Finally, note that
$x_0=y_0$, so reflexivity and
\kl[compatible with multiplication]{compatibility with multiplication}
give $\eta(t_x)\taumg\eta(t_y)$.
\end{proof}

We are now ready to define the decision tree $\calT$ over the block alphabet
$\Sigma^K$:

\begin{definition} %
\label[definition]{def:log-block-decision-tree}
Fix $k\in\N$. Define a \kl{decision tree} $\calT$ over $\Sigma^K$
for words of length $k$ as follows.
For every node $d$ of $\calT$, let $t$ be its \kl{transcript} over
$\Sigma^K$. If no \kl{active interval} of $t$ is nonempty, $d$ is a leaf.
Otherwise, let $N\in\N_+$ be the maximum size of an
\kl{active interval} of $t$.
Choose the leftmost \kl{active interval} of $t$ of size $N$, and let $i$ be the
index of this interval.
Label $d$ by the middle position
$m_i\coloneq\lfloor(s_i+e_i+1)/2\rfloor\in[k]$.
Because $s_i<m_i\le e_i$, we know that position $m_i$ carries $\blank$ in $t$.
Give $d$ one child for every block value in $\Sigma^K$.
Then repeat the same process on the children, until all nodes are leaves:
this terminates as any position queried on a branch is never queried
again on that branch.
\end{definition}

We now show that $\calT$ is 
$\taumg$-coherent and
show the logarithmic bound on its \kl{depth}:

\begin{claim}
\label[claim]{clm:log-tree-correct}
The tree $\calT$ is $\taumg$-coherent.
\end{claim}
\begin{proof}
Let $t$ be
the transcript of a leaf of~$\calT$, and let $t_x,t_y$ be two completions of $t$.
Every active interval of $t$ is empty,
so $t_x$ and $t_y$ agree on all of them. By
\cref{lem:log-active-completions}, $\eta(t_x)\taumg\eta(t_y)$.
\end{proof}

\begin{claimrep}
\label[claim]{clm:log-tree-depth}
The tree $\calT$ has \kl{depth} at most
$B(\lfloor\log_2(k+1)\rfloor+1)$, where $B\coloneq |\Sigma^K|+1$.
\end{claimrep}

\begin{proofsketch}
  We divide each branch into \emph{groups} of $B$ queries and show that
each group halves the largest active-interval size.
Fix the first node $d$ of a group, let $t_d$ be its transcript,
and let $N\coloneq\max\{|I|: I\text{ is an active interval of }t_d\}$
be the largest active-interval size at $d$.
Each query of an interval larger than $N/2$ splits it into
parts of size at most $N/2$, and inactive intervals
remain inactive. As there are at most $B$ active intervals,
within $B$ queries there will be no interval left with size
greater than $N/2$. Thus, after $O(\log k)$ groups, all active intervals are empty.
\end{proofsketch}

\begin{proof}
We partition each branch into groups of at most $B$ queries. We first
show that a query cannot make an unchanged interval active. We then prove
that each group halves the largest active-interval size (measured at the first
node of the group), and iterate this bound.

We first describe how each query changes the intervals. Let $t$ be the
transcript before a query, let $s\in\N$ be its number of non-$\blank$
positions, and let $i\in\{0,\ldots,s\}$ be the index of the
interval containing the queried position. Let $v\in\Sigma^K$ be
the block returned by the query, and let $t'$ be the resulting
transcript. The intervals of $t'$ with indices less than $i$
are the corresponding intervals of $t$. Its $i$-th and $(i+1)$-th
intervals are the two parts of the old $i$-th interval, minus the
queried position. For $j>i+1$, the $j$-th interval of $t'$ is the
$(j-1)$-th interval of $t$.
For each unchanged interval that was inactive in $t$, let $v_j$
be its preceding block and choose
$\ell\in\{j+1,\ldots,s\}$ with $v_\ell=v_j$.
The original non-$\blank$ blocks occur in the same order in $t'$, so
$v_\ell$ is still later than $v_j$. Thus this interval remains
inactive. This means that the active intervals in~$t'$ correspond to active
intervals in~$t$, plus up to two newly active intervals which must be one of the
two parts of the queried interval.

We next prove that the largest active-interval size is halved on a group of $B$
consecutive queries along any branch of the decision tree $\calT$. Fix its first
node $d$ of the group,
let $t_d$ be the transcript of $d$, and let
$N\coloneq\max\{|I|: I\text{ is an active interval of }t_d\}$
be the maximum active-interval size at $d$.
This value is fixed throughout the group. If $N=0$, then $d$
is a leaf and there is nothing to show. 
Otherwise, we say that an active interval is \emph{big} if its size is greater
than~$N/2$.
Initially, at most $B$ active intervals
of~$t_d$ are big, because $t_d$ has at most $B$ active intervals by \cref{clm:log-active-number}.
Now, as we are selecting intervals of maximal size, at any node whose transcript
has at least one big active interval, then we will select a big interval to query.
Now, it suffices to observe that newly active intervals cannot be big: this is because
$d$ had no active intervals of size greater than~$N$, and when querying an interval
of size at most~$N$ then the newly active intervals are obtained by splitting
the queried interval (by the previous paragraph) so no big active interval can be
created (and in particular no active interval of size greater than~$N$ is ever
created).
We have shown that, at each node, at least one big interval is queried and
removed, and no new big intervals are ever created. Thus,
after at most $B$ queries, either the branch ends or every active interval has
size at most $N/2$.

We now iterate along an arbitrary branch. For each group index
$a\in\N$ for which the branch reaches depth $aB$, let $d_a$ be
the node at that depth, and let us define its maximum active-interval size:
$N_a\coloneq\max\{|I|: I\text{ is an active interval of the transcript of
}d_a\}$.
At the root, we have $N_0=k$.
By the preceding paragraph, if both
$d_a$ and $d_{a+1}$ exist, then $N_{a+1}\le N_a/2$. Induction gives
$N_a\le k/2^a$ whenever $d_a$ exists. If $k=0$, the root is a leaf.
Otherwise choose $a\coloneq\lfloor\log_2(k+1)\rfloor+1$. Then
$k/2^a<1$, so any node $d_a$ would have all active intervals empty
and would be a leaf. Thus no branch contains more than $aB$ queries,
and so the claimed depth bound holds.
\end{proof}

As explained after \cref{lem:log-block-query-simulation}, we have now shown
\cref{prop:lang-log-ub} with 
\cref{clm:log-tree-depth,clm:log-tree-correct}.

\begin{toappendix}
  We can now conclude the proof of \cref{prop:lang-log-ub}:
\begin{proof}[Proof of \cref{prop:lang-log-ub}]
Fix $n\in\N$ and let $k\coloneq\lfloor n/K\rfloor$.
For every $1\le j\le k$, construct the tree $\calT_j$ on $j$ blocks
as explained in the main text. By \cref{clm:log-tree-depth}, its depth is
at most $B(\lfloor\log_2(j+1)\rfloor+1)$, where
$B\coloneq|\Sigma^K|+1$ does not depend on $j$.
By \cref{clm:log-tree-correct}, every $\calT_j$ is
$\taumg$-coherent (and the single-leaf tree is coherent at length zero).
Let $d\coloneq B(\lfloor\log_2(k+1)\rfloor+1)$: this bounds the depth of all
these
trees. Since $j\le k$, we have shown that $\taumg$-coherent
query complexity on alphabet $\Sigma^K$ satisfies $C_{\Sigma^K}(k)\le d$.

We now transfer the bound to ordinary query complexity.
Since $\taumg$ is non-shattering, \cref{lem:log-block-query-simulation}
gives $Q_{P,T}(n)\le K\bigl(C_{\Sigma^K}(k)+1\bigr)\le K(d+1)$.
Since $K$ and $B$ are fixed, this is $O(1+\log(n+1))$,
and hence $O(\log n)$ as $n$ grows.
\end{proof}

\end{toappendix}

\subsection{The Contextual Matched-Guard Relations}
\begin{toappendix}
\subsection{The Contextual Matched-Guard Relations}
\end{toappendix}

Having proved the logarithmic upper bound when the \kl{matched-guard tolerance}~$\taumg$
is \kl[relation shatters]{non-shattering}, we now turn to the linear
lower bound when it is \kl[relation shatters]{shattering}.
Similarly to \cref{def:cgrel} in \cref{sec:constant}, we will need to give an
equivalent rephrasing of the
\kl{matched-guard tolerance} which provides a more explicit form for its pairs.
Like \cref{sec:constant} 
we will impose stronger length requirements, but unlike \cref{sec:constant} it
will no longer be sufficient to take one single pair and extend it by one common
left and right context. Instead, we will need to handle up to $s$ pairs at once:

\begin{definition}
\label[definition]{def:matched-guard-decomposition}
For $s\in\N$, define the \intro{$s$-contextual matched-guard relation}
$R_{\cmgrel,s}\subseteq M\times M$ as
\[
R_{\cmgrel,s}\coloneq
\bigcup_{\substack{(k_1,\ldots,k_t)\in\N_+^t\\0\le t\le s}}
\left\{
\left(\begin{array}{l}
a_0(e_1x_1e_1)a_1\cdots(e_tx_te_t)a_t,\\
a_0(e_1y_1e_1)a_1\cdots(e_ty_te_t)a_t
\end{array}\right)
\,\middle|\,
\begin{array}{l}
a_0,\ldots,a_t\in M,\\
e_i\in\Idem(S_\eta),\quad x_i,y_i\in M,\\
e_i,x_i,y_i\in\eta(\Sigma^{k_i})\text{ for all }i\in[t]
\end{array}
\right\}.
\]
\end{definition}

Note that $R_{\cmgrel,0}=\Delta_M$, consisting of pairs of the form $(a_0, a_0)$.
By definition, $R_{\cmgrel,s}\subseteq R_{\cmgrel,s+1}$ for every $s\in\N$,
i.e., the relations 
form an increasing hierarchy.
See \cref{exa:linear-two} in \cref{app:additional-examples} for an illustration
of why taking $s=1$ does not suffice because $R_{\cmgrel,1}$ is not closed under
multiplication.
For the linear lower bound, 
we need to show that the hierarchy collapses at a height independent of the
word length, namely $s\coloneq m^2=|M\times M|$, the number of ordered pairs of elements of the
monoid~$M$:

\begin{lemmarep}
\label[lemma]{lem:matched-guard-closure}
We have
${\taumg} = {\mathrel{R_{\cmgrel,m^2}}}$.
\end{lemmarep}

\begin{proofsketch}
  We prove both inclusions. For $R_{\cmgrel,m^2}\subseteq{\taumg}$, performing
$0$ substitutions gives a diagonal pair, which belongs to $\taumg$
by reflexivity. Otherwise, the common-length condition gives
$e_i\taulen x_i$ and $e_i\taulen y_i$ for each substitution.
Thus each guarded pair $(e_i x_i e_i, e_i y_i e_i)$ belongs to $R_{\mgrel}$.
We conclude via compatibility with multiplication, using reflexivity for the
contexts $a_i$.

For ${\taumg}\subseteq R_{\cmgrel,m^2}$,
\cref{lem:bounded-preimage} expresses each pair in $\taumg$ as a
product of fewer than $m^2$ factors from $R_{\mgrel}\cup\Delta_M$.
We collect the diagonal factors into common contexts.
For each remaining factor $(exe,eye)\in R_{\mgrel}$, the witnesses
for $e\taulen x$ and $e\taulen y$ may have different lengths
$r$ and $s$. Concatenating these witnesses gives preimages of
$e$, $xe$, and $ey$ of the same positive length $r+s$.
Since $e$ is idempotent, $(e(xe)e,e(ey)e)=(exe,eye)$.
We therefore obtain the required decomposition with fewer than
$m^2$ substitutions, so the pair belongs to $R_{\cmgrel,m^2}$.
\end{proofsketch}

\begin{proof}
  We prove ${\taumg}=R_{\cmgrel,m^2}$ by showing the two inclusions, where we recall
that $m$ was defined as $m\coloneq |M|$, i.e., the size of the \kl{monoid} $M$.
First, we prove $R_{\cmgrel,m^2}\subseteq{\taumg}$, which is the easy direction,
simply by establishing that
each pair in the definition of $R_{\cmgrel,m^2}$ is a product of pairs in $\taumg$.
Second, we prove ${\taumg} \subseteq R_{\cmgrel,m^2}$, using the
fact that pairs
of~$\taumg$ can be expressed as a product of at most $m^2$ pairs from $R_{\mgrel}$ and $\Delta_M$.

\medskip

We first prove $R_{\cmgrel,m^2}\subseteq{\taumg}$.
Let $(p,q)\in R_{\cmgrel,m^2}$, where $p,q\in M$.
By \cref{def:matched-guard-decomposition}, there exist an integer $t\in\N$ with $0\le t\le m^2$ and elements $a_0,\ldots,a_t\in M$ which are the contexts.
For each $i\in[t]$, that is, for each integer $1\le i\le t$, the definition gives a positive integer $k_i\in\N_+$, an idempotent $e_i\in\Idem(S_\eta)$, and two elements $x_i,y_i\in M$.
These elements satisfy $e_i,x_i,y_i\in\eta(\Sigma^{k_i})$.
The element $e_i$ is the idempotent (the ``guard''), and $x_i,y_i$ are the
middle elements of the $i$-th pair.
The definition also gives $p=a_0(e_1x_1e_1)a_1\cdots(e_tx_te_t)a_t$ and $q=a_0(e_1y_1e_1)a_1\cdots(e_ty_te_t)a_t$.

If $t=0$, these equalities give $p=q=a_0$, so $(p,q)\in{\taumg}$ by reflexivity.
Suppose now that $t\ge1$.
For each $i\in[t]$, membership of $e_i,x_i,y_i$ in $\eta(\Sigma^{k_i})$ means that each is the image of a word of length $k_i$.
Since $k_i>0$, \cref{def:equal-length-tolerance} gives $e_i\taulen x_i$ and $e_i\taulen y_i$.
Since $e_i\in\Idem(S_\eta)$ and $x_i,y_i\in M$,
\cref{def:matched-guard-tolerance} gives $(e_ix_ie_i,e_iy_ie_i)\in
R_{\mgrel}\subseteq {\taumg}$.
Also, $(a_i,a_i)\in{\taumg}$ for every $0\le i\le t$, by reflexivity of $\taumg$.
By \kl[compatible with multiplication]{compatibility with multiplication}, the product of these pairs belongs to $\taumg$.
Take the product in this order: $(a_0,a_0)$, then $(e_1x_1e_1,e_1y_1e_1)$, then $(a_1,a_1)$, and so on, ending with $(a_t,a_t)$.
The formulas for $p$ and $q$ show that this product is $(p,q)$.
Thus $R_{\cmgrel,m^2}\subseteq\taumg$.

\medskip

We now prove ${\taumg}\subseteq R_{\cmgrel,m^2}$, starting with a bound on the number of factors needed to obtain a pair in $\taumg$.
First note that the relation $R_{\mgrel}$ is symmetric by definition: whenever $(exe,eye)\in R_{\mgrel}$ for $e\in\Idem(S_\eta)$ and $x,y\in M$ with $e\taulen x$ and $e\taulen y$, exchanging $x$ and $y$ gives $(eye,exe)\in R_{\mgrel}$.
However, $R_{\mgrel}$ need not be reflexive, so 
let $\Gamma\coloneq\Delta_M\cup R_{\mgrel}\subseteq M\times M$.
We regard $\Gamma$ as an alphabet whose letters are pairs of elements of $M$.
By \cref{def:tolerance} (more specifically the equivalent rephrasing given
immediately below \cref{def:tolerance}), $\taumg$ is in fact the submonoid of
$M\times M$ generated by $\Gamma$: this uses the fact that $\Gamma$ is already
symmetric because $R_{\mgrel}$ was.

Thus, let us define the \kl[monoid morphism]{morphism}
$\mu\colon\Gamma^*\to{\taumg}$ by $\mu((u,v))\coloneq(u,v)$ for every letter $(u,v)\in\Gamma$.
For a word over $\Gamma$, the \kl[monoid morphism]{morphism} $\mu$ takes the product of its letters in $M\times M$; it maps the empty word to $(1_M,1_M)$.
Since $\Gamma$ generates $\taumg$, this morphism is surjective.
Let $(p,q)\in{\taumg}$, where $p,q\in M$.
Applying \cref{lem:bounded-preimage} to this surjective morphism $\mu$ gives an
integer $r\in\N$ with $r<|{\taumg}|\le|M\times M|=m^2$.
It also gives letters $g_1,\ldots,g_r\in\Gamma$ such that $(p,q)=g_1\cdots g_r$ in $M\times M$.
When $r=0$, this is the empty product $(p,q)=(1_M,1_M)$.

We next show how to replace the diagonal factors in $g_1\cdots g_r$ by contexts, without changing $(p,q)$.
Let $t\in\N$ be the number of factors $g_j$ that do not belong to the \kl{diagonal relation} $\Delta_M$.
Every such factor belongs to $R_{\mgrel}$ because $g_j\in\Gamma=\Delta_M\cup R_{\mgrel}$.
We have $0\le t\le r<m^2$.

If $t=0$, all the factors belong to $\Delta_M$, and their product also belongs to $\Delta_M$.
Thus $p=q$; the case $r=0$ is included here.
Taking $a_0\coloneq p\in M$ and $t=0$ in \cref{def:matched-guard-decomposition} gives $(p,q)=(a_0,a_0)\in R_{\cmgrel,m^2}$.

Suppose now that $t\ge1$.
Write the factors that are not diagonal, in the order in which they occur, as $(p_1,q_1),\ldots,(p_t,q_t)\in R_{\mgrel}$, with $p_i,q_i\in M$ for $i\in[t]$.
We define elements $a_0,\ldots,a_t\in M$ from the diagonal factors.
The element $a_0$ is the product in $M$ of the entries of the diagonal factors before $(p_1,q_1)$.
For $1\le i<t$, the element $a_i$ is the product of the entries of the diagonal factors between $(p_i,q_i)$ and $(p_{i+1},q_{i+1})$.
The element $a_t$ is the product of the entries of the diagonal factors after $(p_t,q_t)$.
These products are taken in $M$, in the original order, so an empty product is $1_M$; the corresponding empty diagonal product in $M\times M$ is $(1_M,1_M)$.
A diagonal factor has the form $(a,a)$ with $a\in M$ and contributes the same element to both coordinates.
Therefore, multiplication in $M\times M$ gives $p=a_0p_1a_1\cdots p_ta_t$ and $q=a_0q_1a_1\cdots q_ta_t$.
We keep the elements $a_0,\ldots,a_t$ as the contexts in \cref{def:matched-guard-decomposition}.

We now show how to choose preimages of the guard and middle elements for each pair $(p_i,q_i)$.
Fix $i\in[t]$.
Since $(p_i,q_i)\in R_{\mgrel}$, \cref{def:matched-guard-tolerance} gives an idempotent $e_i\in\Idem(S_\eta)$ and elements $x_i',y_i'\in M$ such that $p_i=e_ix_i'e_i$, $q_i=e_iy_i'e_i$, $e_i\taulen x_i'$, and $e_i\taulen y_i'$.
Here $e_i$ is the idempotent guard, and $x_i',y_i'$ are the original middle elements.
By \cref{def:equal-length-tolerance}, the relation $e_i\taulen x_i'$ gives a positive integer $r_i\in\N_+$ and words $U_i,V_i\in\Sigma^{r_i}$.
These words satisfy $\eta(U_i)=e_i$ and $\eta(V_i)=x_i'$.
Similarly, $e_i\taulen y_i'$ gives a positive integer $s_i\in\N_+$ and words $U_i',W_i\in\Sigma^{s_i}$.
These words satisfy $\eta(U_i')=e_i$ and $\eta(W_i)=y_i'$.
Thus $U_i$ and $U_i'$ are preimages of the same guard $e_i$, while $V_i$ and
$W_i$ are preimages of the two original middle elements.
The problem is that the lengths $r_i$ and $s_i$ are generally not equal, so we
must now explain how to build preimages of one common positive length to satisfy
the common-length requirements of $R_{\cmgrel,m^2}$.

Thus, let us build words of one common positive length that are preimages of the
guard and of two new middle elements.
Set $k_i\coloneq r_i+s_i\in\N_+$, the required common length.
Define $E_i\coloneq U_iU_i'$, a preimage for the guard; $X_i\coloneq V_iU_i'$, a
preimage for the first new middle element; and $Y_i\coloneq U_iW_i$, a preimage for the second new middle element.
Since $U_i,V_i$ have length $r_i$ and $U_i',W_i$ have length $s_i$, the words $E_i,X_i,Y_i$ all belong to $\Sigma^{k_i}$.
By the morphism property of $\eta$ and $e_i^2=e_i$, we have $\eta(E_i)=\eta(U_i)\eta(U_i')=e_i^2=e_i$, $\eta(X_i)=x_i'e_i$, and $\eta(Y_i)=e_iy_i'$.
Define the new middle elements $x_i\coloneq\eta(X_i)\in M$ and $y_i\coloneq\eta(Y_i)\in M$.
It follows from these three words that $e_i,x_i,y_i\in\eta(\Sigma^{k_i})$.
This is the common-length condition in \cref{def:matched-guard-decomposition}
for the $i$-th pair, so the common-length condition holds for every $i\in[t]$.

We now check that the new middle elements $x_i,y_i$ give the original pair $(p_i,q_i)$.
Since $x_i=x_i'e_i$, $y_i=e_iy_i'$, and $e_i^2=e_i$, we have $e_ix_ie_i=e_i(x_i'e_i)e_i=e_ix_i'e_i=p_i$ and $e_iy_ie_i=e_i(e_iy_i')e_i=e_iy_i'e_i=q_i$.
Thus $(e_ix_ie_i,e_iy_ie_i)=(p_i,q_i)$.
This equality holds for every $i\in[t]$.

We finally prove that $(p,q)\in R_{\cmgrel,m^2}$, completing the second inclusion.
For $t=0$, this was proved above with $a_0=p$.
For $t\ge1$, we have $0<t\le r<m^2$ and contexts $a_0,\ldots,a_t\in M$.
For every $i\in[t]$, we also have a positive integer $k_i\in\N_+$, an idempotent $e_i\in\Idem(S_\eta)$, and middle elements $x_i,y_i\in M$.
They satisfy $e_i,x_i,y_i\in\eta(\Sigma^{k_i})$.
The equalities for $(p_i,q_i)$ and the products with the contexts give $p=a_0(e_1x_1e_1)a_1\cdots(e_tx_te_t)a_t$ and $q=a_0(e_1y_1e_1)a_1\cdots(e_ty_te_t)a_t$.
These are exactly the conditions of \cref{def:matched-guard-decomposition} with $s=m^2$.
Consequently, $(p,q)\in R_{\cmgrel,m^2}$.
Since $(p,q)\in{\taumg}$ was arbitrary, ${\taumg}\subseteq R_{\cmgrel,m^2}$.
Together with the first inclusion, this proves ${\taumg}=R_{\cmgrel,m^2}$.
\end{proof}

\subsection{Linear Lower Bound}
\label{sec:linlb}
\begin{toappendix}
\subsection{Linear Lower Bound}
\end{toappendix}

We now turn to the linear lower bound of \cref{thm:target-trichotomy},
showing that a \kl[relation shatters]{shattering} \kl{matched-guard tolerance}
forces at least linear \kl{query complexity}:

\begin{propositionrep}[Linear lower bound]
\label[proposition]{prop:lang-linear-lb}
If the \kl{matched-guard tolerance} $\taumg$ is \kl[relation shatters]{shattering}, then $Q_{P,T}(n)$ is in~$\Omega(n)$.
\end{propositionrep}

\begin{toappendix}
  We first show the auxiliary results given in the main text:
\end{toappendix}

We prove this in the rest of the section. We first illustrate our
criterion on $(P_{\lin},T_{\lin})$ (\cref{tab:examples}):

\begin{example}
\label[example]{exa:linear}
Let $\Sigma=\{a,c,d\}$, $P_{\lin}\coloneq a^*(c\mid d)a^*$,
and $T_{\lin}\coloneq a^*ca^*$ (as in \cref{tab:examples}).
Let $(M,\eta,M_P,M_T)$ be any \kl{mutual presentation} of this pair.
Take the integer $\omega\coloneq |M|!\in\N_+$ as an \kl{idempotent power} of $M$.
Put $e\coloneq\eta(a^\omega)$, $x\coloneq\eta(a^{\omega-1}c)$,
and $y\coloneq\eta(a^{\omega-1}d)$, all in $M$.
Then $e\in\Idem(S_\eta)$, and $e\taulen x$ and $e\taulen y$.
The witnesses are respectively the word pairs
$(a^\omega,a^{\omega-1}c)$ and $(a^\omega,a^{\omega-1}d)$,
each consisting of words of length $\omega$.

Thus $(exe,eye)\in R_{\mgrel}$ is a \kl{shattered pair},
as witnessed by $a^{2\omega-1}ca^\omega\in T_{\lin}$ and
$a^{2\omega-1}da^\omega\in P_{\lin}\setminus T_{\lin}$.
This shows that the \kl{matched-guard tolerance} is \kl[relation shatters]{shattering}.
\end{example}

Let us now show the linear lower bound for $P_{\lin}$, $T_{\lin}$:

\begin{lemmarep}[Linear lower bound]
\label[lemma]{lem:linear-block-words}
The pair $(P_{\lin},T_{\lin})$ of \cref{exa:linear} has
$Q_{P_{\lin},T_{\lin}}(n)$ in $\Omega(n)$.
\end{lemmarep}

In fact the \kl{query complexity} of
$(P_{\lin},T_{\lin})$ is precisely $\Theta(n)$, since querying every position
gives a linear upper bound. For now, let us sketch the lower bound:

\begin{proofsketch}
Fix $n\in\N_+$ and answer every query with $a$. Before all $n$ positions are
queried, some $i\in[n]$ remains unknown, so both $a^{i-1}ca^{n-i}\in T_{\lin}$
and $a^{i-1}da^{n-i}\in P_{\lin}\setminus T_{\lin}$ complete the
\kl{transcript}. A \kl{correct} tree therefore cannot decide before \kl{depth} $n$.
\end{proofsketch}

\begin{proof}
Fix a word length $n\in\N_+$, and let $\calT$ be a \kl{correct} \kl{decision tree} for $(P_{\lin},T_{\lin})$ from \cref{exa:linear} at length $n$. We show that $\calT$ has depth at least $n$.

For every position $i\in[n]$, consider the word of $T_{\lin}$ having $c$ at position $i$ and $a$ everywhere else, and the word of $P_{\lin}\setminus T_{\lin}$ having $d$ at position $i$ and $a$ everywhere else. More precisely, define $u_i\coloneq a^{i-1}ca^{n-i}$ and $v_i\coloneq a^{i-1}da^{n-i}$. Then $u_i\in T_{\lin}$ and $v_i\in P_{\lin}\setminus T_{\lin}$.

We follow a path in $\calT$ by answering every query with the letter $a$. Suppose that after fewer than $n$ queries the tree reaches a leaf. Since fewer than $n$ positions have been queried, there is some position $i\in[n]$ which has not been queried.

Both $u_i$ and $v_i$ are then \kl[completion]{completions} of the \kl{transcript} of this leaf. Indeed, every queried position is different from $i$, and both words contain the letter $a$ at every such position. However, $u_i\in T_{\lin}$ whereas $v_i\in P_{\lin}\setminus T_{\lin}$. Hence the \kl{transcript} of this leaf is \kl[partial word shatters]{shattering}.

By \cref{lem:insensitive-leaf}, the \kl{transcript} of a leaf of a \kl{correct} \kl{decision tree} cannot be \kl[partial word shatters]{shattering}.
Therefore the path obtained by answering every query with $a$ cannot reach a leaf before all $n$ positions have been queried.

Consequently, $\calT$ has depth at least $n$. Since $\calT$ was an arbitrary \kl{correct} \kl{decision tree} at length $n$, every \kl{correct} \kl{decision tree} for $(P_{\lin},T_{\lin})$ at length $n$ has depth at least $n$. Therefore $Q_{P_{\lin},T_{\lin}}(n)\geq n$, and in particular $Q_{P_{\lin},T_{\lin}}(n)$ is in~$\Omega(n)$.
\end{proof}

To prove \cref{prop:lang-linear-lb}, as in the previous section,
we will reduce from \cref{lem:linear-block-words}. We will use for
this the rephrasing of the matched-guard tolerance as the contextual
matched-guard relation, along with
the closure lemmas of \cref{sec:constant}
(inclusion, quotient, and inverse length-multiplying morphism).
However,
because the contextual matched-guard relation may perform multiple substitutions
at once (up to $m^2$), we will additionally need a form of closure under
\emph{concatenation}:

\begin{lemmarep}[Concatenation]
\label[lemma]{lem:query-concatenation}
Let $s\in\N_+$, and for each $i\in[s]$ let
$T_i\subseteq P_i\subseteq\Sigma^*$. Define
$T=T_1\cdots T_s$ and
$P=T\cup(P_1\setminus T_1)\cdots(P_s\setminus T_s)$.
Assume that $T\cap(P_1\setminus T_1)\cdots(P_s\setminus T_s)=\varnothing$.
Then, for every $n\in\N$, we have
$Q_{P,T}(sn)\ge\min_{i\in[s]}Q_{P_i,T_i}(n)$.
\end{lemmarep}

\begin{toappendix}
  For this, we first show the auxiliary result:
\end{toappendix}

The disjointness assumption in the lemma ensures that concatenating
target words cannot also give a concatenation of non-target words.
The proof of \cref{lem:query-concatenation} is more complex than that of the
analogous results in \cref{sec:constant}. It uses the notion of the \emph{cost}
of a partial word to keep track of how many queries remain to be done:

\begin{definition}[Cost of a partial word]
\label[definition]{def:partial-word-cost}
Let $T\subseteq P\subseteq\Sigma^*$, let $n\in\N$, and let
$w\in(\Sigma\cup\{\blank\})^n$ be a \kl{partial word}.
Write $L(w)\coloneq\{u\in\Sigma^n\mid u\text{ is a \kl{completion} of }w\}$.
The \intro{cost} of $w$ relative to $(P,T)$ is
$\operatorname{cost}_{P,T}(w)\coloneq Q_{P\cap L(w),\,T\cap L(w)}(n)$.
\end{definition}
Thus $\operatorname{cost}_{P,T}(w)>0$ exactly when $w$
\kl[partial word shatters]{shatters} $(P,T)$, by
\cref{lem:insensitive-leaf}. Further, for the all-$\blank$ partial words, we have
$Q_{P,T}(n)=\max_{0\le\ell\le n}
\operatorname{cost}_{P,T}(\blank^\ell)$.
The next lemma shows that an answer to a query can always be chosen
to decrease the cost by at most one.

\begin{lemmarep}[Cost preservation]
\label[lemma]{lem:cost-preservation}
Let $T\subseteq P\subseteq\Sigma^*$, let $w$ be a \kl{partial word}
of length $n\in\N$, and let $j\in[n]$ such that $w_j=\blank$.
There is a letter $a\in\Sigma$ such that the \kl{partial word} $w'_a$
obtained by placing $a$ at position $j$ satisfies
$\operatorname{cost}_{P,T}(w'_a)\ge\operatorname{cost}_{P,T}(w)-1$.
\end{lemmarep}

\begin{proofsketch}
  If the cost of~$w_a'$ were smaller than the bound across every $a\in \Sigma$,
  as witnessed by decision trees
  $\calT_a$, then a decision tree for~$w$ with root labeled $j$ and child subtrees
  $\calT_a$ would have height less than the cost of~$w$, a contradiction.
\end{proofsketch}

\begin{proof}
Let $k\coloneq\operatorname{cost}_{P,T}(w)\in\N$. For every $a\in\Sigma$, let
$w_a'$ be the \kl{partial word} obtained by placing $a$ at
position $j$, and let $d_a\coloneq\operatorname{cost}_{P,T}(w_a')\in\N$.
Choose a \kl{correct} \kl{decision tree} $\calT_a$ of minimum
\kl{depth} $d_a$ for the promised \kl[completion]{completions} of
$w_a'$, i.e., for the language pair
$P\cap L(w_a')$ and $T\cap L(w_a')$.
If $\calT_a$ queries a position whose letter is specified
by $w_a'$, we can replace this query by the subtree reached after
that letter. Doing so for every such query preserves \kl[correct]{correctness}
on the promised completions and does not increase the depth.
In particular, the resulting tree never queries position $j$.

We construct a \kl{decision tree} which first queries $j$ and,
for every possible answer $a\in\Sigma$, follows the corresponding
resulting tree. No position is queried twice on a branch, and every
promised completion of $w$ for the answer $a \in \Sigma$ follows the tree
$\calT_a$.
Thus the constructed tree is \kl{correct} for
$(P\cap L(w),T\cap L(w))$ at length $n$, and has \kl{depth} at
most $1+\max_{a\in\Sigma}d_a$. By the definition of \kl{cost},
$k\le1+\max_{a\in\Sigma}d_a$. Hence there is some $a\in\Sigma$
with $d_a\ge k-1$, and choosing $w'=w_a'$ proves the claim.
\end{proof}

In particular, when starting from a \kl{partial word} of cost $k\in\N$,
we can answer any $q\in\N$ successive queries so that the cost
remains at least $k-q$. When $q<k$, the resulting \kl{partial word}
is still \kl[partial word shatters]{shattering}. This is used to prove
\cref{lem:query-concatenation} in the appendix, intuitively by choosing a branch where we keep
track of the cost of the different parts of the concatenation.
\begin{toappendix}
  We can now conclude the proof of \cref{lem:query-concatenation}:
\end{toappendix}

\begin{toappendix}
\begin{proof}[Proof of \cref{lem:query-concatenation}]
  We will in fact rephrase the desired result to the following stronger claim
  using the notion of cost (*):
  for any partial words $w_1, \ldots, w_s$
  letting
$d\coloneq\min_{i\in[s]}\operatorname{cost}_{P_i,T_i}(w_i)$, we have
$d \leq \operatorname{cost}_{P,T}(w_1 \cdots w_s)$.
If $d=0$, the cost inequality is immediate. Otherwise let $\calT$
be a correct decision tree for
$(P\cap L(w_1\cdots w_s),T\cap L(w_1\cdots w_s))$
at length $|w_1\cdots w_s|$, and
we show that $\calT$ has depth at least $d$ to establish (*).

We will now inductively define a specific \kl{branch} $\beta = d^{(0)}, \ldots$ of the tree $\calT$ from its root to a leaf. In this paragraph, we describe exactly how to choose this branch. In the next paragraph, we show that it has length at least $d$, witnessing the desired lower bound on the \kl{depth} of~$\calT$.
For this, we will examine the transcripts of nodes along this branch, towards the goal of
applying \cref{lem:cost-preservation} on them.
We initialize $\beta$ with $d^{(0)}$ being the root of the tree, and will now
use the transcripts to explain how we choose the node $d^{(q+1)}$ from the
preceding node $d^{(q)}$ for each $q \geq 0$ until the branch has reached a leaf.
If $d^{(q)}$ is a leaf, the definition of~$\beta$ is finished.
Otherwise, considering the internal node $d^{(q)}$,
let $w_i^{(q)}$ for each $i \in [s]$ be the partial word obtained as the restriction of the
transcript of $d^{(q)}$ to the factor
corresponding to $w_i$ together with the letters already specified in $w_i$: 
we initially have $w_i^{(0)}=w_i$.
Now, there are two cases. First, 
if $d^{(q)}$ queries a position which is already known in the initial partial
word $w = w_1 \cdots w_s$, then the only possible answer is the one
corresponding to~$w$ and we extend the
definition of~$\beta$ by setting $d^{(q+1)}$ to be the child of~$d^{(q)}$ that corresponds to the
only possible answer: we then have 
$w_i^{(q+1)} = w_i^{(q)}$. 
Second, if $d^{(q)}$ queries a position which carries a~$\blank$ in the initial partial
word $w = w_1 \cdots w_s$, it falls into a range corresponding to one factor,
say $w_i$ for $i \in [s]$. We will then have
$w_j^{(q+1)} = w_j^{(q)}$ for each $j \in [s]$ with $j \neq i$, and 
$w_i^{(q+1)}$ will be obtained from~$w_i^{(q)}$ by revealing one additional letter
(i.e., replacing one $\blank$) at the position that $d^{(q)}$ queries.
We are now in a position to apply \cref{lem:cost-preservation}
for the pair $(P_i,T_i)$, the partial word $w_i^{(q)}$, and the queried position
within the part of the unknown word corresponding to~$i$.
This gives us a letter $a\ \in \Sigma$, and we choose the child of $d^{(q)}$ labeled
$a$ to be $d^{(q+1)}$, extending the definition of the branch~$\beta$.

We now show that the branch $\beta$ thus defined cannot finish before depth $d$,
establishing the lower bound on the depth of~$\calT$.
Suppose by way of contradiction that $\beta$
reaches a leaf $d^{(q)}$ after $q<d$ queries. The inequality 
of \cref{lem:cost-preservation} shows that each $w_i^{(q)}$ has positive cost,
because we had 
$\operatorname{cost}_{P_i,T_i}(w_i^{(0)}) \geq d$ by our choice of~$d$ and
because $w_i^{(0)} = w_i$, and for each of the $q$ queries along~$\beta$ we
either had $w_i^{(q+1)} = w_i^{(q)}$ or 
$w_i^{(q+1)}$ obtained from~$w_i^{(q)}$ by choosing a letter achieving the
inequality of \cref{lem:cost-preservation}, i.e., decrementing the cost by at
most~$1$, so we conclude by $q<d$.
As each $w_i^{(q)}$ has positive cost, we know that it shatters $(P_i,T_i)$.
Thus, choose completions $u_i\in T_i$ and $v_i\in P_i\setminus T_i$
of $w_i^{(q)}$ for every $i\in[s]$.
Then $u_1\cdots u_s\in T$ and
$v_1\cdots v_s\in(P_1\setminus T_1)\cdots(P_s\setminus T_s)$.
We assumed that these two sets are disjoint, so $v_1 \cdots v_s$ is in $P\setminus T$.
Both words complete $w_1\cdots w_s$ and reach leaf $d^{(q)}$,
so this witnesses that the transcript of $d^{(q)}$ is shattering,
contradicting the correctness of the tree by
\cref{lem:insensitive-leaf}. Thus the depth of $\calT$ is at least
$d$, proving
$\operatorname{cost}_{P,T}(w_1\cdots w_s)\ge d$.

Having shown (*), we explain why this implies the result claimed in the
statement.
Fix $n\in\N$.
We have 
$Q_{P_i,T_i}(n)=\max_{0\le r\le n}
\operatorname{cost}_{P_i,T_i}(\blank^r)$, so choose
$n_i\le n$ for each $i\in[s]$ as a value of~$r$ that attains the maximum.
Set $\ell\coloneq\sum_{i=1}^s n_i\le sn$.
We have $Q_{P,T}(sn)=\max_{0\le r\le sn}
\operatorname{cost}_{P,T}(\blank^r)$, so in particular
$Q_{P,T}(sn)\ge\operatorname{cost}_{P,T}(\blank^\ell)$, we have
$\operatorname{cost}_{P,T}(\blank^\ell)\ge
\min_{i\in[s]}\operatorname{cost}_{P_i,T_i}(\blank^{n_i})$ by our stronger claim
(*), and we have 
$Q_{P_i,T_i}(n_i)=
\operatorname{cost}_{P_i,T_i}(\blank^{n_i})$
for each $i \in [s]$ by our choice
of~$n_i$.
Combining the three inequalities proves the claim.
\end{proof}
\end{toappendix}

We now apply the inclusion, quotient, and inverse-morphism reductions
of \cref{sec:constant} to a shattering guarded pair in the form of 
\cref{lem:matched-guard-closure}. Together with the concatenation lemma
(\cref{lem:query-concatenation}), this allows us to deduce
\cref{prop:lang-linear-lb} from the lower bound on $P_{\lin},T_{\lin}$
(\cref{lem:linear-block-words}):

\begin{proofsketch}
Since $\taumg$ is \kl[relation shatters]{shattering}, the equality
${\taumg}=R_{\cmgrel,m^2}$ from \cref{lem:matched-guard-closure}
gives a \kl{shattered pair} with at most $m^2$ substitutions.
Let $t\in\N_+$ be the number of non-identity substitutions.
For each $i\in[t]$, choose equal-length preimages
$E_i,X_i,Y_i\in\Sigma^{k_i}$, for some $k_i\in\N_+$,
for its guard and middle elements.
Choose fixed words $b_0,\ldots,b_t\in\Sigma^*$ for the contexts.
By idempotence, the languages
$T_i\coloneq E_i^+X_iE_i^+b_i$ and
$P_i\coloneq T_i\cup E_i^+Y_iE_i^+b_i$, for $i\in[t]$,
give languages $T'\coloneq b_0T_1\cdots T_t\subseteq T$ and
$P'\coloneq T'\cup b_0(P_1\setminus T_1)\cdots(P_t\setminus T_t)$
with $P'\setminus T'\subseteq P\setminus T$.

By the inclusion lemma (\cref{lem:query-inclusion}), $(P,T)$ is at least as hard as $(P',T')$.
By the quotient lemma (\cref{lem:query-fixed-concatenation}) we can remove the fixed prefix $b_0$.
We then use the concatenation lemma (\cref{lem:query-concatenation}) and the inverse length-multiplying
morphism lemma (\cref{lem:query-inverse-morphism})
to argue that it suffices to show lower bounds on the pairs
$(P_i,T_i)$. 
Each of them is handled by reducing from the fixed language pair $P_{\lin},T_{\lin}$
(using \cref{lem:linear-block-words})
via the inclusion and quotient lemmas.
\end{proofsketch}

\begin{toappendix}
\begin{proof}[Proof of \cref{prop:lang-linear-lb}]
Since $\taumg$ is \kl[relation shatters]{shattering},
choose a \kl{shattered pair} $(p,q)\in{\taumg}$ with
$p\in M_T$ and $q\in M_P\setminus M_T$. In particular, $p\ne q$.
By \cref{lem:matched-guard-closure}, we have
$(p,q)\in R_{\cmgrel,m^2}$.
Thus, by \cref{def:matched-guard-decomposition}, there exist
$t\in\N$ with $t\le m^2$, elements $a_0,\ldots,a_t\in M$,
and, for every $i\in[t]$, an idempotent
$e_i\in\Idem(S_\eta)$, elements $x_i,y_i\in M$, and
$k_i\in\N_+$ such that
$p=a_0(e_1x_1e_1)a_1\cdots(e_tx_te_t)a_t$ and
$q=a_0(e_1y_1e_1)a_1\cdots(e_ty_te_t)a_t$, with
$e_i,x_i,y_i\in\eta(\Sigma^{k_i})$ for every $i\in[t]$.

We may assume $(\ast)$: $e_ix_ie_i\ne e_iy_ie_i$ for all $i\in[t]$.
Indeed, if $e_ix_ie_i=e_iy_ie_i$, replace the consecutive
factors $a_{i-1}(e_ix_ie_i)a_i$ and
$a_{i-1}(e_iy_ie_i)a_i$ by their common product in both
coordinates and omit the $i$-th substitution.
After reindexing, the expressions for $p,q$, the bound $t\le m^2$,
and the common-length conditions of the remaining substitutions
are preserved. This proves $(\ast)$. Further, at least one substitution
remains, because if $t=0$ then the defining expressions give
$p=q=a_0$ but we have $p\ne q$. Thus, we must have $t\ge1$.

For every $i\in[t]$, choose
$E_i,X_i,Y_i\in\Sigma^{k_i}$ with
$\eta(E_i)=e_i$, $\eta(X_i)=x_i$, and $\eta(Y_i)=y_i$.
For every $i\in\{0,\ldots,t\}$, choose $b_i\in\Sigma^*$
with $\eta(b_i)=a_i$; these words exist by surjectivity of $\eta$.

For $i\in[t]$, define
$T_i\coloneq E_i^+X_iE_i^+b_i$ and
$P_i\coloneq T_i\cup E_i^+Y_iE_i^+b_i$.
For all $r,s\in\N_+$, the idempotence of $e_i$ gives
$\eta(E_i^rX_iE_i^s)=e_ix_ie_i$ and
$\eta(E_i^rY_iE_i^s)=e_iy_ie_i$.
These two elements of $M$ are different, so
$E_i^+X_iE_i^+$ and $E_i^+Y_iE_i^+$ are disjoint.
Appending the same fixed word $b_i$ preserves disjointness.
Consequently $P_i\setminus T_i=E_i^+Y_iE_i^+b_i$.

Define $T'\coloneq b_0T_1\cdots T_t$ and
$P'\coloneq T'\cup b_0(P_1\setminus T_1)\cdots(P_t\setminus T_t)$.
By $\eta(E_i^rX_iE_i^s)=e_ix_ie_i$ and
$\eta(E_i^rY_iE_i^s)=e_iy_ie_i$ for all $r,s\in\N_+$,
and by $\eta(b_i)=a_i$,
every word in $T'$ has image $p$, whereas every word in
$b_0(P_1\setminus T_1)\cdots(P_t\setminus T_t)$ has image $q$.
Since $p\in M_T$ and $q\in M_P\setminus M_T$,
these two languages are disjoint, further we have $T'\subseteq T$ and
$P'\setminus T'\subseteq P\setminus T$.
By the inclusion lemma (\cref{lem:query-inclusion}), we deduce that
$Q_{P,T}(N)\ge Q_{P',T'}(N)$ for every $N\in\N$, so it suffices to show the lower bound for $P',T'$.

Now, removing the fixed prefix $b_0$, define
$T''\coloneq b_0^{-1}T'=T_1\cdots T_t$ and
$P''\coloneq b_0^{-1}P'=T''\cup
(P_1\setminus T_1)\cdots(P_t\setminus T_t)$.
Since $T'=b_0T''$ and $P'=b_0P''$, the quotient lemma
(\cref{lem:query-fixed-concatenation}) gives
$Q_{P',T'}(N+|b_0|)\ge Q_{P'',T''}(N)$ for every $N\in\N$.
So it suffices to show the lower bound for $P'',T''$.

Fix an integer $k\ge3$ and, for each $i\in[t]$, set
$n_i\coloneq k_i k+|b_i|\in\N$.
We show that the query complexity for  $(P_i,T_i)$
at length $n_i$ is at least $k-2$.
Fix $i\in[t]$. Since $E_i,X_i,Y_i$ all have length $k_i$,
removing the fixed first block $E_i$ and the fixed suffix $E_i b_i$
gives
$T_i'\coloneq E_i^{-1}T_i(E_i b_i)^{-1}=E_i^*X_iE_i^*$ and
$P_i'\coloneq E_i^{-1}P_i(E_i b_i)^{-1}=T_i'\cup E_i^*Y_iE_i^*$.
Indeed, after removing $b_i$, every word in $P_i$ is a
concatenation of blocks of length $k_i$ beginning and ending
with $E_i$. Deleting these two blocks leaves precisely a word
of $E_i^*X_iE_i^*$ or of $E_i^*Y_iE_i^*$, respectively,
and every word in either language arises in this way.
In particular, $T_i=E_iT_i'E_i b_i$ and
$P_i=E_iP_i'E_i b_i$.
The two languages $E_i^*X_iE_i^*$ and $E_i^*Y_iE_i^*$
are also disjoint. Otherwise, for a common word $z$, the word
$E_i z E_i$ would belong to both $E_i^+X_iE_i^+$ and
$E_i^+Y_iE_i^+$, contradicting their disjointness.

Define the morphism $h_i\colon\{a,c,d\}^*\to\Sigma^*$ by
$h_i(a)\coloneq E_i$, $h_i(c)\coloneq X_i$, and $h_i(d)\coloneq Y_i$.
All three images have length $k_i\in\N_+$.
By the definitions of $(P_{\lin},T_{\lin})$ in
\cref{lem:linear-block-words}, we have the following inclusions (**):
$h_i(T_{\lin})\subseteq T_i'$ and
$h_i(P_{\lin}\setminus T_{\lin})\subseteq
E_i^*Y_iE_i^*\subseteq P_i'\setminus T_i'$.
Thus the map $w\mapsto E_ih_i(w)E_i b_i$ sends
$T_{\lin}$ into $T_i$ and
$P_{\lin}\setminus T_{\lin}$ into $P_i\setminus T_i$.

We now transfer the lower bound of \cref{lem:linear-block-words} through the
three reduction lemmas on each factor of the concatenation.
We know by \cref{lem:linear-block-words} that 
$k-2 \le Q_{P_{\lin},T_{\lin}}(k-2)$. 
The two inclusions for $h_i$ given in (**) imply, by the inclusion lemma
(\cref{lem:query-inclusion}),
that
$Q_{P_{\lin},T_{\lin}}(k-2)\le
Q_{h_i^{-1}(P_i'),\,h_i^{-1}(T_i')}(k-2)$.
By the inverse lm-morphism lemma (\cref{lem:query-inverse-morphism}),
this last quantity is at most
$Q_{P_i',T_i'}(k_i(k-2))$. We can now apply the quotient lemma
(\cref{lem:query-fixed-concatenation}) with prefix $E_i$ and suffix
$E_i b_i$ to show an upper bound on the last quantity given by
$Q_{P_i,T_i}(k_i(k-2)+2k_i+|b_i|)=Q_{P_i,T_i}(n_i)$.
Chaining all inequalities together, we get
$Q_{P_i,T_i}(n_i)\ge k-2$ for every $i\in[t]$.

We now combine these bounds using the concatenation lemma (\cref{lem:query-concatenation}).
Put $n_*\coloneq\max_{i\in[t]}n_i$.
We verify the disjointness hypothesis of
\cref{lem:query-concatenation}. If a word $z$ belonged to both
$T''$ and $(P_1\setminus T_1)\cdots(P_t\setminus T_t)$,
then $b_0z$ would have both images $p$ and $q$ under $\eta$,
contradicting $p\ne q$. Hence
$Q_{P'',T''}(t n_*)\ge\min_{i\in[t]}Q_{P_i,T_i}(n_*)\ge k-2$,
where the second inequality uses $n_*\ge n_i$ and monotonicity.

Finally, we transfer the bound to $(P,T)$.
The quotient lemma
(\cref{lem:query-fixed-concatenation})
and inclusion lemma (\cref{lem:query-inclusion}) give
$Q_{P,T}(|b_0|+t n_*)\ge k-2$.
Put $C\coloneq\max_{i\in[t]}k_i\in\N_+$ and
$D\coloneq\max_{i\in[t]}|b_i|\in\N$.
Since $n_i=k_i k+|b_i|\le Ck+D$ for each $i$,
we have $Q_{P,T}(|b_0|+t(Ck+D))\ge k-2$
for every integer $k\ge3$.
As the lower bound holds for infinitely many values,
by monotonicity of $Q_{P,T}$, we finally have
that $Q_{P,T}(n)$ is in~$\Omega(n)$.
\end{proof}
\end{toappendix}

\section{Effective Complexity Classification}\label{sec:effective}

We have concluded the proof of \cref{thm:target-trichotomy} which classifies
possible query complexity regimes.
In this section, we study how we can decide this regime when given finite
descriptions of $P$ and $T$ as \intro[nondeterministic finite
automaton]{nondeterministic finite automata}
(NFAs), for which we use standard definitions~\cite[Chapter~2]{HopcroftMotwaniUllman2001}.

\begin{theoremrep}[Effective classification]
\label[theorem]{thm:effective-classification}
Let $T\subseteq P\subseteq\Sigma^*$ be \kl[regular language]{regular languages}.
Given \kl[nondeterministic finite automaton]{NFAs} $A_P$ and $A_T$ recognizing
$P$ and $T$, respectively, it is PSPACE-complete to decide which of the three regimes of \cref{thm:target-trichotomy} applies to $Q_{P,T}$.
\end{theoremrep}

The PSPACE upper bound therefore also applies when $P$ and $T$ are described by
regular expressions, because they can be 
translated to NFAs in polynomial time~\cite[Section~3.2.3]{HopcroftMotwaniUllman2001}.
The proof uses standard techniques for implicit representations of the transition monoid of input automata:

\begin{proofsketch}
  The challenge to prove the upper bound is that we cannot explicitly build a
  \kl{mutual presentation} for $P$ and~$T$, because its size can be exponential
  in the NFA. Instead, we recall classical constructions on NFAs to show how we
  can implicitly represent the transition monoid of an NFA, and thus a mutual
  presentation for our two NFAs. This relies on the fact that elements of the
  transition monoid can be represented as polynomial-sized matrices, and that
  while there may be exponentially many possible products we have  
  bounds on the maximal useful length of a product (via
  \cref{lem:bounded-preimage,lem:equal-length-tolerance}) which ensure that
  polynomially many bits for the counter suffice. Having an implicit
  representation of the mutual presentation is then enough to test our
  criteria through standard techniques, in particular using nondeterminism thanks to Savitch's theorem~\cite{Savitch1970}.
  The lower bound reduces from NFA universality using standard techniques
  similar to Greibach's theorem~\cite{greibach1968note}.
\end{proofsketch}

\begin{proof}[Proof of \cref{thm:effective-classification}]
  We first show the PSPACE upper bound, before showing PSPACE-hardness at the
  end of the section.

Let $A_P$ and $A_T$ be the input NFAs which accept $P$ and $T$ respectively.
Write $Q_P,Q_T$ for their respective state sets, 
of respective cardinalities $n_P,n_T\in\N$. 
We use the standard notion of \intro[Boolean matrix]{Boolean
matrices}~\cite[Chapter~IV, Section~3.2]{Pin2025} to represent elements of the
transition monoid of~$A_P$ and~$A_T$.
For an NFA with state set $Q$, the \kl[Boolean matrix]{Boolean matrix} $A_w\in\{0,1\}^{Q\times Q}$ of $w\in\Sigma^*$ has $(A_w)_{p,q}=1$ exactly when a path from $p\in Q$ to $q\in Q$ is labeled $w$. For two such matrices $A,B$, define their product by $(AB)_{p,q}=\bigvee_{r\in Q}(A_{p,r}\wedge B_{r,q})$. Thus $A_{uv}=A_uA_v$ and $A_\varepsilon$ is the identity matrix. The map $w\mapsto A_w$ is the \intro{transition-relation morphism}.

For each word $w\in\Sigma^*$, we write
$A_w^P$ and $A_w^T$ for the \kl[Boolean matrix]{Boolean matrices} corresponding
to~$w$ via the \kl{transition-relation morphism}, and we
define $\eta(w)\coloneq(A_w^P,A_w^T)$ and $M\coloneq\eta(\Sigma^*)$, with componentwise multiplication. 
Define $M_P\subseteq M$ as the set of matrix pairs whose first matrix has an entry $1$ from an initial state to a final state of the automaton for $P$. Define $M_T\subseteq M$ analogously using the second matrix and the automaton for $T$. 
It is then easy to notice that we have
$P=\eta^{-1}(M_P)$ and $T=\eta^{-1}(M_T)$. As $T\subseteq P$ and $\eta$ is surjective onto $M$, we have $M_T\subseteq M_P$. Hence $(M,\eta,M_P,M_T)$ is a \kl{mutual presentation}.
Put $s\coloneq n_P^2+n_T^2\in\N$ and $B\coloneq 2^s\in\N_+$. Each element of $M$
is represented by $s$ bits, and we have $|M|\le B$. 
Note that we do not construct the mutual presentation explicitly.
Each individual element of~$M$ can be represented in polynomial space (as a pair
of transition matrices), the
transition matrix pairs for the letters of the alphabet are given in the input, and
given two pairs we can test them for equality, test them for membership
in~$M_T$ and~$M_P$, and multiply them in polynomial
space by multiplying their two constituent matrices. Idempotence can also be tested in polynomial space.

By Savitch's theorem~\cite{Savitch1970}, NPSPACE equals PSPACE. It therefore suffices to give nondeterministic polynomial-space algorithms.
We first explain how to test membership in $M$, in~$S_\eta$, and in~$\taulen$.
Given a pair of matrices $z$, we can test whether it represents an element
of~$M$ and whether it belongs to~$S_\eta$ in the
following way.
We start with the identity pair $\eta(\varepsilon)$, and we guess a word letter by letter, updating its image by multiplying by $\eta(a)$ for each guessed letter $a\in\Sigma$.
We accept if the final image equals $z$, and maintain a polynomially-sized
counter and reject if the sequence of guesses exceeds $B$.
(Note that we do not store the guessed sequence explicitly, because it may have
exponential length.)
This correctly decides 
membership in $M$ by \cref{lem:bounded-preimage}, since every element of $M$ has a preimage of length less than $|M|\le B$.
For membership in $S_\eta$, use the same procedure, but require at least one
letter.
The process is also correct because non-empty preimages witnessing membership
to~$S_\eta$ must also have length at most~$B$. Indeed, take a shortest nonempty
preimage. If its length exceeds $|M|$, two of its first $|M|+1$ nonempty prefix
images coincide, and removing the factor between these prefixes gives a shorter nonempty preimage with the same image, a contradiction.

To test $x\taulen y$ for two candidate matrix pairs $x,y$, start with two
identity pairs. At each step, guess two letters $a,b\in\Sigma$ and multiply the
respective running images by $\eta(a)$ and $\eta(b)$.
Guessing the two words in parallel ensures that they have equal length.
Accept when the two images
are equal to $x$ and $y$ and at least one step has been performed, and reject if the
process has not succeeded in $|M|^2\le B^2$ steps: 
this suffices by \cref{lem:equal-length-tolerance}.

The membership tests for $R_{\grel}$ and $R_{\mgrel}$ are similar, so we present both of them at once for a generic $R\in\{R_{\grel},R_{\mgrel}\}$.
To test whether two given matrix pairs are related by $R=R_{\grel}$, guess any
choice of candidate matrix pairs $e,f,x,y$ and check $e,f\in S_\eta$, $e^2=e$, $f^2=f$, $x\taulen y$, and the two equalities defining $(exf,eyf)$. For $R=R_{\mgrel}$, guess matrix pairs $e,x,y$ and check $e\in S_\eta$, $e^2=e$, $e\taulen x$, $e\taulen y$, and the equalities defining $(exe,eye)$. The preceding tests verify 
in particular that the guessed elements actually belong to $M$. Thus membership in $R$, in its converse $R^{-1}\coloneq\{(y,x)\mid(x,y)\in R\}$, or in the diagonal $\Delta_M\coloneq\{(z,z)\mid z\in M\}$ can be tested in nondeterministic polynomial space. 

Lastly for the upper bound, we explain how to test membership of a pair of candidate matrix pairs in the
\kl{tolerance closure} $C$ of $R \in \{R_{\grel},R_{\mgrel}\}$.
Put $G\coloneq R\cup R^{-1}\cup\Delta_M$. By the bound established after \cref{def:tolerance}, every pair in $C$ is a product of fewer than $|M|^2\le B^2$ elements of $G$. Guess such a product by starting at $(1_M,1_M)$ and then guessing each factor, verifying its membership to~$G$ by
the preceding tests, accumulating the product of the guessed elements, 
and stopping if we exceed $B^2$ steps, which again only requires a polynomially sized counter.
This decides membership in $C$.

Thus, we can test membership in $\taug$ and $\taumg$ in nondeterministic polynomial space. Last, to test whether $C$ \kl[relation shatters]{shatters} $(M_P,M_T)$, guess two matrix pairs, we simply test that one belongs to $M_T$ and the other to $M_P\setminus M_T$, and apply the membership test to check if the pair of the two matrix pairs belongs to $C$. 
Thus, both tolerance membership and shattering are decidable in PSPACE, which establishes that we can determine
the regime of \cref{thm:target-trichotomy} in PSPACE.

\medskip

We last show the PSPACE lower bound, which adapts techniques from Greibach's
theorem~\cite{greibach1968note}. The reduction is from the problem of NFA
universality, which is well-known to be PSPACE-hard. Let $A$ be an input NFA on
which we want to decide universality, and let $\Sigma$ be its alphabet. Let $\#,
a, c, d$ be four fresh letters.
We build in polynomial time from~$A$ two NFAs $A_P$ and $A_T$
recognizing the promise language $P \coloneq (\Sigma^*) \# a^* (c|d) a^*$
and the target language
$T \coloneq (\mathrm{L}(A) \# a^* (c|d) a^*) | (\Sigma^* \# a^* c a^*)$,
where $\mathrm{L}(A)$ denotes the language accepted by~$A$.
Note that $T \subseteq P$.
We now claim the following, which suffices to conclude:
\begin{itemize}
  \item If $A$ is universal, then conditional membership to $P,T$ can be done
    in~$O(1)$ (in fact without making any queries at all);
  \item If $A$ is not universal, then conditional membership to~$P,T$ requires
    $\Theta(n)$ queries.
\end{itemize}
Indeed, if $A$ is universal, then we have $\mathrm{L}(A) = \Sigma^*$, so that
the first term of the union in the definition of~$T$ is a superset of the
second, so we can rewrite equivalently $T = \Sigma^* \# a^* (c|d) a^*$ and we have in
fact $P = T$. Thus, conditional membership for $P,T$ is trivial and can be solved with no
queries at all.

By contrast, if $A$ is not universal, there is a word $w \in \Sigma^* \setminus
\mathrm{L}(A)$. We now consider the languages $P' \coloneq (w\#)^{-1} P$ and $T'
\coloneq (w\#)^{-1} T$. We have $P' = a^* (c|d) a^*$ and $T' = a^* c a^*$,
noting in particular that the first term of the union in the definition of~$T$
contributes no terms to the quotient. We recognize that $P', T'$  is in fact
precisely $P_{\lin},T_{\lin}$ which has a linear query complexity lower bound by
\cref{lem:linear-block-words}, and we know that taking quotients preserve lower
bounds by \cref{lem:query-fixed-concatenation}. Thus, conditional membership
for~$P,T$ requires $\Omega(n)$ queries.
This establishes correctness and
concludes the proof.
\end{proof}

\section{Conclusion}\label{sec:conclusion}

In this paper, we have studied the query complexity of the \kl{conditional membership} problem for
regular languages, i.e., determining whether a word $w$ belongs to a fixed
target~$T$ under the
hypothesis that $w$ belongs to a promise~$P$.
Our main theorem gives
the three possible asymptotic regimes: $O(1)$ (where we can
intuitively decide by querying the left and right endpoints of the word),
$\Theta(\log n)$ (where we can intuitively decide by performing a
generalization of binary search), and $\Theta(n)$ (where we must examine
linearly many positions).
We have further shown that the criterion in our
trichotomy is PSPACE-complete to decide given automata representations of our
two languages $P$ and~$T$.

Our work opens several directions for further research. 
One first question would be to translate our algebraic criteria to 
language-theoretic conditions on $P$ and $T$ (independent
from the mutual presentation). Another question is computational complexity: we
do not know if
the PSPACE upper bound can be lowered
for $P$ and $T$ given as DFAs. 
We also do not know the combined complexity of our
upper bound algorithms when $P$ and $T$ are given as input.
To improve them, we could in particular avoid the Ramsey-type argument
from \cite{Jecker2021}  to find idempotent subfactors;
we may be able to use traditional pumping arguments instead, at the cost of complicating the
proofs. This could ensure block sizes that are linear in the mutual
presentation, though this is still superexponential in the size of DFAs
representing $P$ and $T$ (see, e.g., \cite{holzer2004deterministic}).

Another natural question is to fix a regular promise $P$ and determine which of
the three complexity regimes can occur across all possible regular
targets $T\subseteq P$. 
This would generalize the constant/linear dichotomy for 
non-conditional membership~\cite{AaronsonGrierSchaeffer2019}, 
which corresponds to $P = \Sigma^*$, and we suspect that it would
relate to~\cite{Moshkov2000}.
We also leave open the question of the \emph{nondeterministic} 
or \emph{randomized} query complexity in our models. For \emph{randomized} query complexity, 
we suspect our methodology goes through so the bounds would be unchanged. By contrast, the 
nondeterministic case is different:
for $P_{\log},T_{\log}$ or even
$P_{\lin},T_{\lin}$
from \cref{tab:examples}
the
nondeterministic
query complexity is constant because one single observed letter suffices to answer
positively. We suspect that the constant class
matches with $\Sigma_1[<, +1, \text{Mod}]$ of langages defined by the existential
fragment of first order logic with order and modular predicates~\cite{Chaubard06}. 

\nosectionappendix
\begin{toappendix}
\section{Additional Examples}\label[appendix]{app:additional-examples}

This section collects examples illustrating successive binary searches in
\cref{prop:lang-log-ub} and the need for multiple substitutions in
\cref{def:matched-guard-decomposition}.

\begin{example}
\label[example]{exa:log-two}
Let $\Sigma=\{a,b,c,d,e\}$, $T_{\mathrm{log}2}\coloneq a^*cb^*ce^*$, and $P_{\mathrm{log}2}\coloneq a^*(c\mid d)b^*(c\mid d)e^*$. Here, finding only the boundary between the $a$ and $b$ regions is not enough: we also need to find the boundary between the $b$ and $e$ regions. This corresponds to performing two binary searches, with each letter indicating in which region we fall.

\end{example}

\begin{example}
\label[example]{exa:log-blocks}
Let $\Sigma=\{a,b,c,d\}$ be the alphabet, $T_{\mathrm{log}3}\coloneq (aab)^*c(aab)^*c(aab)^*$, and $P_{\mathrm{log}3}\coloneq (aab)^*(c\mid d)(aab)^*(c\mid d)(aab)^*$. 
This is a variant of the previous example, where querying blocks of size~$3$ indicates in which $(aab)^*$ factor we are, depending on the remainder of the position of the~$b$ modulo~$3$.
\end{example}

\begin{example}
\label[example]{exa:log-nested}
Let $\Sigma=\{a,b,c,d,e,f\}$, $T_{\mathrm{log}4}\coloneq e^*(a^*cb^*\mid b^*ca^*)f^*$, and $P_{\mathrm{log}4}\coloneq e^*(a^*(c\mid d)b^*\mid b^*(c\mid d)a^*)f^*$. Here we illustrate a nested binary search. All $a$'s precede all $b$'s or vice versa, so until both an $a$ and a $b$ have been queried, an answer $a$ alone does not tell us which direction to search. Alternatively, a first binary search can locate the end of the initial $e$ region: querying the next position either reveals $c$ or $d$ and settles membership, or reveals $a$ or $b$ and determines the order. Thus, the direction indicated by an answer in the inner binary search may only become clear once we have completed the outer binary search.
\end{example}

\begin{example}
\label[example]{exa:linear-two}
Let $\Sigma=\{a,b,c,d\}$, $T_{\mathrm{lin}2}\coloneq a^*ca^*b^*cb^*$ and $P_{\mathrm{lin}2}\coloneq T_{\mathrm{lin}2}\cup a^*da^*b^*db^*$. The \kl{promise} requires the two occurrences from $\{c,d\}$ to agree: both are $c$, or both are $d$.
Changing only one occurrence of $c$ to $d$ violates the \kl{promise}, so the two replacements must be coordinated.

Take the \kl{mutual presentation} that identifies two words when they have the same membership in both languages within every left and right context.
Identify each word with its image in the monoid $M$.
Then $a^2=a$ and $b^2=b$, since repeating an $a$ or $b$ does not change membership in either language.
The pairs $(aca,ada)$ and $(bcb,bdb)$ belong to $R_{\cmgrel,1}$: the guards are respectively $a$ and $b$, and all guards and middle words have length $1$.
Their product is $(acabcb,adabdb)$.

This product does not belong to $R_{\cmgrel,1}$.
Indeed, one replacement must change both $c$'s to $d$'s, so its two copies of the same nonempty guard must surround both positions.
The left copy would consist only of $a$'s and the right copy only of $b$'s, which is impossible for the same nonempty word.
Thus $R_{\cmgrel,1}$ is not \kl{compatible with multiplication}.
\end{example}
\end{toappendix}

\bibliographystyle{plainurl}
\bibliography{conditional-membership-icalp2027}

\begin{thebibliography}{10}

\bibitem{AaronsonGrierSchaeffer2019}
Scott Aaronson, Daniel Grier, and Luke Schaeffer.
\newblock A quantum query complexity trichotomy for regular languages.
\newblock In {\em FOCS}, 2019.
\newblock \href {https://doi.org/10.1109/FOCS.2019.00061}
  {\path{doi:10.1109/FOCS.2019.00061}}.

\bibitem{AlonKrivelevichNewmanSzegedy2001}
Noga Alon, Michael Krivelevich, Ilan Newman, and Mario Szegedy.
\newblock Regular languages are testable with a constant number of queries.
\newblock {\em {SIAM} Journal on Computing}, 30(6), 2001.
\newblock \href {https://doi.org/10.1137/S0097539700366528}
  {\path{doi:10.1137/S0097539700366528}}.

\bibitem{AmarilliJachietPaperman2021}
Antoine Amarilli, Louis Jachiet, and Charles Paperman.
\newblock Dynamic membership for regular languages.
\newblock In {\em ICALP}, volume 198, 2021.
\newblock \href {https://doi.org/10.4230/LIPIcs.ICALP.2021.116}
  {\path{doi:10.4230/LIPIcs.ICALP.2021.116}}.

\bibitem{AmarilliLabbePaperman2026}
Antoine Amarilli, Sebastien Labbe, and Charles Paperman.
\newblock {Out-of-order} membership in regular languages.
\newblock In {\em ICALP}, volume 374, 2026.
\newblock \href {https://doi.org/10.4230/LIPIcs.ICALP.2026.161}
  {\path{doi:10.4230/LIPIcs.ICALP.2026.161}}.

\bibitem{BCST92}
David A.~Mix Barrington, Kevin~J. Compton, Howard Straubing, and Denis
  Th{\'{e}}rien.
\newblock Regular languages in {NC}{\({^1}\)}.
\newblock {\em J. Comput. Syst. Sci.}, 44(3), 1992.
\newblock \href {https://doi.org/10.1016/0022-0000(92)90014-A}
  {\path{doi:10.1016/0022-0000(92)90014-A}}.

\bibitem{BathieFijalkowMascle2025}
Gabriel Bathie, Nathana\"{e}l Fijalkow, and Corto Mascle.
\newblock The trichotomy of regular property testing.
\newblock In {\em ICALP}, volume 334, 2025.
\newblock \href {https://doi.org/10.4230/LIPIcs.ICALP.2025.141}
  {\path{doi:10.4230/LIPIcs.ICALP.2025.141}}.

\bibitem{BuhrmanDeWolf2002}
Harry Buhrman and Ronald de~Wolf.
\newblock Complexity measures and decision tree complexity: A survey.
\newblock {\em Theoretical Computer Science}, 288(1), 2002.
\newblock \href {https://doi.org/10.1016/S0304-3975(01)00144-X}
  {\path{doi:10.1016/S0304-3975(01)00144-X}}.

\bibitem{Chaubard06}
Laura Chaubard, Jean-Éric Pin, and Howard Straubing.
\newblock First order formulas with modular predicates.
\newblock In {\em LICS}, 2006.
\newblock \href {https://doi.org/10.1109/LICS.2006.24}
  {\path{doi:10.1109/LICS.2006.24}}.

\bibitem{chen2009learning}
Yu-Fang Chen, Azadeh Farzan, Edmund~M Clarke, Yih-Kuen Tsay, and Bow-Yaw Wang.
\newblock Learning minimal separating {DFA}'s for compositional verification.
\newblock In {\em International Conference on Tools and Algorithms for the
  Construction and Analysis of Systems}, 2009.

\bibitem{dell2024dfaminer}
Daniele Dell’Erba, Yong Li, and Sven Schewe.
\newblock Dfaminer: Mining minimal separating {DFA}s from labelled samples.
\newblock In {\em International Symposium on Formal Methods}, 2024.

\bibitem{Eilenberg1976}
Samuel Eilenberg.
\newblock {\em Automata, languages, and machines. Volume B}.
\newblock Academic Press, 1976.

\bibitem{DBLP:journals/theoretics/GanardiHLMS25}
Moses Ganardi, Danny Hucke, Markus Lohrey, Konstantinos Mamouras, and Tatiana
  Starikovskaya.
\newblock Regular languages in the sliding window model.
\newblock {\em TheoretiCS}, 4, 2025.
\newblock URL: \url{https://doi.org/10.46298/theoretics.25.8}, \href
  {https://doi.org/10.46298/THEORETICS.25.8}
  {\path{doi:10.46298/THEORETICS.25.8}}.

\bibitem{geffert2015classical}
Viliam Geffert and Abuzer Yakaryilmaz.
\newblock Classical automata on promise problems.
\newblock {\em Discrete Mathematics \& Theoretical Computer Science},
  17(Automata, Logic and Semantics), 2015.

\bibitem{goldreich2006promise}
Oded Goldreich.
\newblock On promise problems: A survey.
\newblock In {\em Theoretical Computer Science: Essays in Memory of Shimon
  Even}. 2006.

\bibitem{goralvcik1986discerning}
Pavel Goral{\v{c}}{\'\i}k and V{\'a}clav Koubek.
\newblock On discerning words by automata.
\newblock In {\em ICALP}, 1986.

\bibitem{greibach1968note}
Sheila Greibach.
\newblock A note on undecidable properties of formal languages.
\newblock {\em Mathematical systems theory}, 2(1), 1968.

\bibitem{holzer2004deterministic}
Markus Holzer and Barbara K{\"o}nig.
\newblock On deterministic finite automata and syntactic monoid size.
\newblock {\em Theoretical Computer Science}, 327(3), 2004.

\bibitem{HopcroftMotwaniUllman2001}
John~E. Hopcroft, Rajeev Motwani, and Jeffrey~D. Ullman.
\newblock {\em Introduction to automata theory, languages, and computation}.
\newblock Addison-Wesley, 2nd edition, 2001.

\bibitem{Jecker2021}
Isma{\"e}l Jecker.
\newblock A {R}amsey theorem for finite monoids.
\newblock In {\em STACS}, volume 187, 2021.
\newblock \href {https://doi.org/10.4230/LIPIcs.STACS.2021.44}
  {\path{doi:10.4230/LIPIcs.STACS.2021.44}}.

\bibitem{martens2023trichotomy}
Wim Martens, Matthias Niewerth, and Tina Popp.
\newblock A trichotomy for regular trail queries.
\newblock {\em Logical Methods in Computer Science}, 19, 2023.

\bibitem{Moshkov2000}
Mikhail Moshkov.
\newblock Decision trees for regular language word recognition.
\newblock {\em Fundamenta Informaticae}, 41(4), 2000.
\newblock \href {https://doi.org/10.3233/FI-2000-41405}
  {\path{doi:10.3233/FI-2000-41405}}.

\bibitem{Pin1986}
Jean-{\'{E}}ric Pin.
\newblock {\em Varieties of formal languages}.
\newblock North Oxford Academic, 1986.

\bibitem{Pin2025}
Jean-{\'E}ric Pin.
\newblock Mathematical foundations of automata theory, 2025.
\newblock Version of March 24, 2025.
\newblock URL: \url{https://www.irif.fr/~jep/PDF/MPRI/MPRI.pdf}.

\bibitem{PlaceZeitoun2016}
Thomas Place and Marc Zeitoun.
\newblock Separating regular languages with first-order logic.
\newblock {\em Logical Methods in Computer Science}, 12(1), 2016.
\newblock \href {https://doi.org/10.2168/LMCS-12(1:5)2016}
  {\path{doi:10.2168/LMCS-12(1:5)2016}}.

\bibitem{Radelecki1991}
S{\'a}ndor Radelecki.
\newblock Compatible tolerances on groupoids.
\newblock {\em Czechoslovak Mathematical Journal}, 41(3), 1991.
\newblock \href {https://doi.org/10.21136/CMJ.1991.102478}
  {\path{doi:10.21136/CMJ.1991.102478}}.

\bibitem{Savitch1970}
Walter~J. Savitch.
\newblock Relationships between nondeterministic and deterministic tape
  complexities.
\newblock {\em Journal of Computer and System Sciences}, 4(2), 1970.
\newblock \href {https://doi.org/10.1016/S0022-0000(70)80006-X}
  {\path{doi:10.1016/S0022-0000(70)80006-X}}.

\bibitem{Steinberg2016}
Benjamin Steinberg.
\newblock {\em Representation theory of finite monoids}.
\newblock Springer, Cham, 2016.
\newblock \href {https://doi.org/10.1007/978-3-319-43932-7}
  {\path{doi:10.1007/978-3-319-43932-7}}.

\bibitem{TessonTherien2005}
Pascal Tesson and Denis Th{\'{e}}rien.
\newblock Complete classifications for the communication complexity of regular
  languages.
\newblock {\em Theory of Computing Systems}, 38(2), 2005.
\newblock \href {https://doi.org/10.1007/s00224-004-1190-2}
  {\path{doi:10.1007/s00224-004-1190-2}}.

\end{thebibliography}

\end{document}